\documentclass[letterpaper,USenglish, pdfa, thm-restate]{lipics-v2021}

\pdfoutput=1 
\hideLIPIcs  

\newcommand{\scc}{\textsc{SCC}}
\newcommand{\seth}{\textsc{SETH}}

\newcommand{\ksum}{$k$-\textsc{SUM}}
\newcommand{\kxor}{$k$-\textsc{XOR}}
\newcommand{\ppseth}{$\pw$\textsc{-SETH}}
\newcommand{\twseth}{$\tw$\textsc{-SETH}}
\newcommand{\pw}{\textrm{pw}}
\newcommand{\tw}{\textrm{tw}} 
 
\newcommand{\cw}{\textrm{cw}}

\newcommand{\eps}{\varepsilon}

\title{k-SUM Hardness Implies Treewidth-SETH} 

\author{Michael Lampis}{Universit\'{e} Paris-Dauphine, PSL University, CNRS UMR7243, LAMSADE, Paris, France}{michail.lampis@dauphine.fr}{https://orcid.org/0000-0002-5791-0887}{}

\authorrunning{M. Lampis} 

\Copyright{Michael Lampis}



\category{} 

\relatedversion{} 

\funding{This work was supported by ANR project ANR-21-CE48-0022 (S-EX-AP-PE-AL).}

\nolinenumbers 

\EventEditors{John Q. Open and Joan R. Access}
\EventNoEds{2}
\EventLongTitle{42nd Conference on Very Important Topics (CVIT 2016)}
\EventShortTitle{CVIT 2016}
\EventAcronym{CVIT}
\EventYear{2016}
\EventDate{December 24--27, 2016}
\EventLocation{Little Whinging, United Kingdom}
\EventLogo{}
\SeriesVolume{42}
\ArticleNo{23}

\usepackage{todonotes}

\usepackage{tikz}

\begin{document}

\maketitle

\begin{abstract}

We show that if \ksum\ is hard, in the sense that the standard algorithm is
essentially optimal, then a variant of the \seth\ called the Primal Treewidth
\seth\ is true.  Formally: if there is an $\eps>0$ and an algorithm which
solves \textsc{SAT} in time $(2-\eps)^{\tw}|\phi|^{O(1)}$, where $\tw$ is the
width of a given tree decomposition of the primal graph of the input, then
there exists a randomized algorithm which solves \ksum\ in time
$n^{(1-\delta)\frac{k}{2}}$ for some $\delta>0$ and all sufficiently large $k$.
We also establish an analogous result for the \kxor\ problem, where integer
addition is replaced by component-wise addition modulo $2$. 

An interesting aspect of our proof is that we rely on two key ideas from
different topics. First, inspired by the classical perfect hashing scheme of
Fredman, Koml{\'{o}}s, and Szemer{\'{e}}di, we show that \ksum\ admits an
interactive proof protocol using integers of absolute value only $O(n^{k/2})$.
Second, using the intuition that SAT formulas of treewidth $\tw$ can encode the
workings of alternating Turing machines using $\tw$ bits of space, we are able
to encode this protocol into a formula of treewidth roughly $\frac{k}{2}\log n$
and obtain the main result. 

As an application of our reduction we are able to revisit tight lower bounds on
the complexity of several fundamental problems parameterized by treewidth
(\textsc{Independent Set}, \textsc{Max Cut}, $k$-\textsc{Coloring}). Our
results imply that these bounds, which were initially shown under the \seth,
also hold if one assumes the \ksum\ or \kxor\ Hypotheses, arguably increasing
our confidence in their validity.

\end{abstract}

\newpage

\tableofcontents

\newpage

\section{Introduction}\label{sec:intro}

\subsection{Executive summary}

The main result of this paper is to establish the following implication: if
there exists a faster-than-expected algorithm for solving \textsc{SAT}
parameterized by treewidth, then there exists a faster-than-expected
(randomized) algorithm for \ksum\ and \kxor. Equivalently, the \ksum\
Hypothesis implies the \twseth.

The immediate application (and initial motivation) of this work is in
parameterized complexity, where our results establish that several dynamic
programming (DP) algorithms for standard problems (\textsc{Independent Set},
\textsc{Max Cut}, $k$-\textsc{Coloring}) parameterized by treewidth are optimal
\emph{under the} \ksum\ Hypothesis; this is new evidence for these important
lower bounds, which were previously known under the \seth\ and its variants.
Beyond this immediate application, our result is, to the best of our knowledge,
one of the first non-trivial implications reducing a \textsc{Sum}-based
hypothesis to a \textsc{SAT}-based hypothesis. This is valuable in the wider
line of research attempting to draw connections among the numerous hypotheses
in fine-grained complexity and offers some new information about which
hypotheses should be considered more solid.

An interesting aspect of our result is that we obtain it by combining
intuitions from complexity theory (about space-bounded non-deterministic or
alternating machines) and perfect hashing schemes. Our starting point is a
natural question about \ksum: what relation should one expect to have between
the number of given integers ($n$) and their values? It is known that one can
assume without loss of generality that the maximum absolute value is at most
(roughly) $O(n^k)$ and this can be achieved via hashing. Could one go further?

We expect the answer to this question to be negative, barring some very clever
idea. However, we work around this difficulty and hash an arbitrary \ksum\
instance into one where all integers have absolute values $O(n^{k/2})$ by
\emph{cheating} in the sense that we change the rules of the question. In
particular, we consider an interactive proof setting where a verifier
challenges a prover to produce a solution in a hashed instance with small
absolute values, and then the verifier is allowed to pose a follow-up challenge
to eliminate false positives. The key idea is then that, because \textsc{SAT}
parameterized by treewidth captures the workings of non-deterministic machines
with bounded space and a small number of alternations, this protocol can be
encoded in a \textsc{SAT} formula of treewidth (roughly) $\frac{k}{2}\log n$.
Hence, our proof is an example of how intuitions from complexity theory can
help us make non-trivial connections between problems which initially look
quite different.

\subsection{Background} 

The investigation of fine-grained complexity hypotheses and the algorithmic
lower bounds they imply has been a topic of intense study in the last two
decades.  The three most prominent such hypotheses are probably the Strong
Exponential Time Hypothesis (\seth), the $3$-\textsc{SUM} Hypothesis, and the
\textsc{All-Pairs Shortest Paths} (\textsc{APSP})
Hypothesis\footnote{Subsequently to the initial publication of this article,
Alman and Vassilevska Williams \cite{alman2026} announced an algorithm
disproving the $3$-\textsc{SUM} and \textsc{APSP} hypotheses. The introduction
of this paper has been slightly updated to take into account these
developments, which at the moment do not seem to significantly impact the
results of this paper.   }, but numerous other plausible hypotheses have also
been proposed in the literature. One drawback of this line of work is the fact
that most hypotheses do not have a clear relation to each other, making it
harder to judge their relative solidity.   

%

Our goal in this paper is to attempt to ameliorate this situation in a
direction that is of special interest for a specific well-studied sub-field of
fine-grained complexity, namely the fine-grained complexity of
\emph{structurally parameterized} problems. In this area, numerous results are
known proving that dynamic programming algorithms used in conjunction with
standard parameters, such as treewidth, are optimal under the \seth.  In
particular, the pioneering work of Lokshtanov, Marx, and Saurabh
\cite{LokshtanovMS18} started a line of work which succeeded in attacking
problems for which an algorithm running in time $c^{\tw}n^{O(1)}$ was known,
for $c$ a constant, and showing that an algorithm with running time
$(c-\eps)^{\tw}n^{O(1)}$ would falsify the \seth\footnote{Here $\tw$ denotes
the width of a given tree decomposition of the input graph. Throughout the
paper we assume the reader is familiar with the basics of parameterized
complexity, as given in standard textbooks such as \cite{CyganFKLMPPS15}.}.
Results of this type have been obtained for virtually all the major problems
that admit single-exponential algorithms parameterized by treewidth, as well as
for other parameters, such as pathwidth, clique-width, and cutwidth (we review
some results below).

Although this line of research has been quite successful, it was somewhat
undesirable that the only evidence given for these numerous lower bounds was
the \seth, a hypothesis whose veracity is not universally accepted. This
motivated a recent work of Lampis \cite{Lampis25} where it was shown that many
(perhaps most) of these lower bounds can be obtained from a weaker, more
plausible assumption, called the Primal Pathwidth \seth\ (\ppseth). One of the
main advantages of this approach is that the \ppseth\ is more credible, as it
was shown to be implied by several assumptions other than the \seth.
Nevertheless, this line of work still made no connection between
\textsc{SAT}-based hypotheses and hypotheses from the two other main families
(\textsc{Sum}-based and \textsc{APSP}-based) and hence all evidence for the
optimality of DP algorithms parameterized by treewidth still mostly rests on
variants of the \seth.

\subsection{Our results} 

In this paper we make what is, to the best of our knowledge, one of the first such
connections between a \textsc{Sum}-based and a \textsc{SAT}-based fine-grained
hypothesis and as our main application obtain new evidence for several (known)
fine-grained lower bounds for standard problems parameterized by treewidth. In
a nutshell, our main result is that a stronger variant of the $4$-\textsc{SUM}
Hypothesis implies a (weaker, but still quite useful) variant of the \seth.   

In order to be more precise, let us define exactly the relevant hypotheses,
starting with two \textsc{Sum}-related hypotheses.  The \ksum\ problem is the
following: we are given $k$ arrays (for some fixed $k$), each containing $n$
integers and the goal is to decide if it is possible to select one integer from
each array so that the sum of the selected integers is $0$. It is a well-known
fact that a meet in the middle approach can solve this problem in time
$\widetilde{O}(n^{\lceil\frac{k}{2}\rceil})$\footnote{The $\widetilde{O}$
notation suppresses factors poly-logarithmic in the argument.} and the research
for faster algorithms has been extremely active
\cite{BaranDP08,Chan20,ChanL15,GoldS17,GronlundP18}. For $k=3$ this gives a
quadratic algorithm. It was long conjectured that this is optimal, until very
recent work announced by Alman and Vassilevska Williams \cite{alman2026}
disproved this hypothesis. Nevertheless, even this recent breakthrough does not
supply us with an algorithm that is truly faster than meet-in-the-middle.
Intuitively, what is shown is that in the case of $k=3$ the exponent can be
strictly smaller than $\left\lceil \frac{k}{2}\right\rceil$; however, the
bigger question is (arguably) whether the exponent can be strictly smaller that
$\frac{k}{2}$, which seems much more challenging algorithmically. In
particular, Alman and Vassilevska Williams explicitly state that an algorithm
solving $4$-\textsc{SUM} in time $O(n^{2-\eps})$ still seems out of reach.

Given the above, we consider the following formulation of our main hypothesis:
The \ksum\ Hypothesis states that the meet in the middle algorithm is almost
optimal, meaning that for all fixed $k\ge 3$ and $\eps>0$ it is impossible to
solve the problem in time $O(n^{\frac{k}{2}-\eps})$ (even when allowing
randomized algorithms).  Observe that since the \ksum\ Hypothesis is stated for
all $k\ge 3$ it implies the $4$-\textsc{SUM} Hypothesis, which is the
hypothesis that $4$-\textsc{SUM} does not admit sub-quadratic algorithms, and
which is (as we mentioned) unaffected by the results of \cite{alman2026}. 

A close variant of these problems is the \kxor\ problem, where rather than
summing integers we are summing boolean vectors (component-wise modulo $2$).
The meet in the middle algorithm is similarly conjectured to be optimal and
this is referred to as the \kxor\ Hypothesis. It is currently unknown whether
the \kxor\ Hypothesis is stronger or weaker than the \ksum\ Hypothesis
\cite{DalirrooyfardLS25}.  Note that for both problems it is potentially
relevant to specify the number of bits necessary to encode the given integers
or the dimension of the given vectors.  As we will see, it will in general be
safe to assume that we are considering elements encoded with $O(k\log n)$ bits,
but the hidden constant will be very relevant.

Let us also recall some relevant \textsc{SAT}-related hypotheses. The
aforementioned \ppseth\ states that the ``obvious'' dynamic programming
algorithm for solving \textsc{SAT} parameterized by pathwidth is optimal. More
precisely, it states that for all $\eps>0$, it is impossible to decide if a CNF
formula $\phi$ given together with a path decomposition of its primal graph of
width $\pw$ and size\footnote{Of course, in all reasonable cases we can safely
assume that the size of the given decomposition is at most polynomial in
$|\phi|$. We state this requirement explicitly to avoid degenerate cases where
the algorithm is supplied with a decomposition containing a huge number of
redundant bags.} polynomial in $|\phi|$ is satisfiable in time $(2-\eps)^\pw
|\phi|^{O(1)}$.  Observe that the \seth\ trivially implies the \ppseth.
Replacing path by tree decompositions in the statement of the \ppseth\ gives us
the \twseth, which was first explicitly posited by Iwata and Yoshida
\cite{IwataY15}.  Since every path decomposition is a tree decomposition, the
\ppseth\ trivially implies the \twseth. Unlike the \seth, it is known that for
the \ppseth\ and \twseth\ the arity of the given CNF formula is not important,
namely, the hypotheses are equivalent whether we state them for \textsc{SAT} or
$3$-\textsc{SAT}

\smallskip

Having set the groundwork, we are now ready to informally state the main result
of this paper.

\begin{theorem}[Informal]

The following statements hold:

\begin{enumerate}

\item The \ksum\ hypothesis implies the \twseth.

\item The \kxor\ hypothesis implies the \twseth.

\end{enumerate}

\noindent More strongly, we have the following: if the \twseth\ is false, then
there exists an $\eps>0$ such that for sufficiently large $k$, both \ksum\ and
\kxor\ admit randomized algorithms running in time $n^{(1-\eps)\frac{k}{2}}$.

\end{theorem}

More precise statements corresponding to the main theorem are given in
Theorems~\ref{thm:kxor} and \ref{thm:ksum}.  A summary of the relations between
several hypotheses, where our result is placed into context, is given in
Figure~\ref{fig:summary}.

On a high level, the heart of our main result is a randomized reduction which
transforms a given \ksum\ or \kxor\ instance into a CNF formula of treewidth
$\frac{k}{2}\log n$. The technical part of this paper is devoted to describing
this transformation and proving that it works correctly with high probability.

It is worth noting here that the \ksum\ Hypothesis is usually stated in a quite
careful way, where the exponent is $\lceil\frac{k}{2}\rceil-\eps$, that is, it
makes a difference whether $k$ is odd or even. In contrast to this, we
establish that if the \twseth\ were false, then we would strongly falsify the
\ksum\ Hypothesis, in the sense that we would be able to improve upon the
standard algorithm not only by an additive but by a multiplicative constant
(albeit, for sufficiently large $k$). As a result, in the remainder of the
paper it will not be crucial whether $k$ is odd or even (and for practical
purposes, we will usually assume it is even and sufficiently large).

\begin{figure}[h]

\input{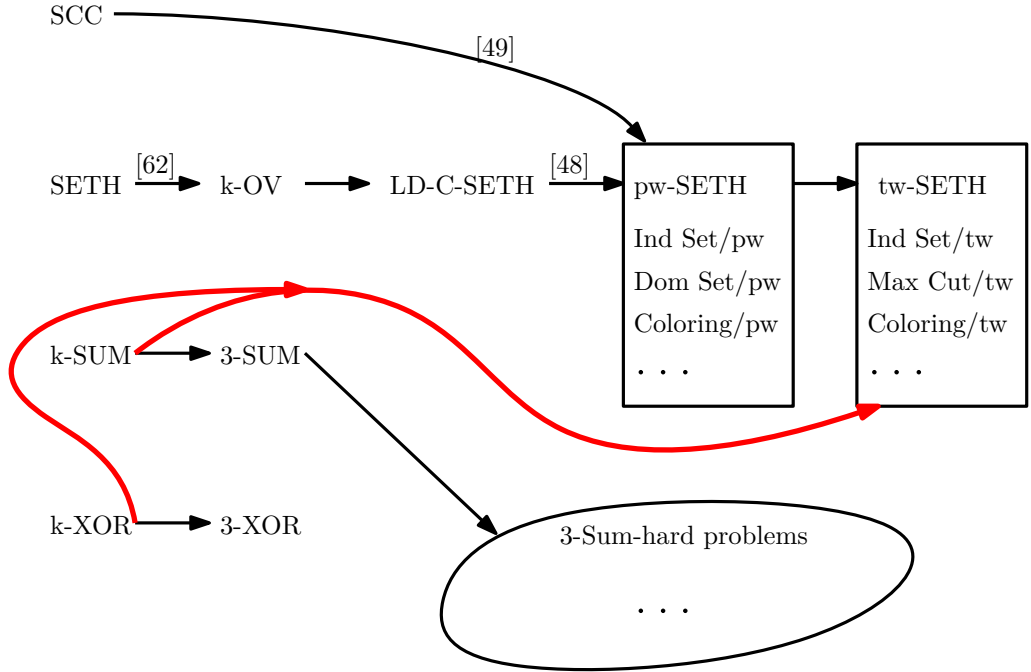} \tikzstyle{ipe stylesheet} = [
  ipe import,
  even odd rule,
  line join=round,
  line cap=butt,
  ipe pen normal/.style={line width=0.4},
  ipe pen heavier/.style={line width=0.8},
  ipe pen fat/.style={line width=1.2},
  ipe pen ultrafat/.style={line width=2},
  ipe pen normal,
  ipe mark normal/.style={ipe mark scale=3},
  ipe mark large/.style={ipe mark scale=5},
  ipe mark small/.style={ipe mark scale=2},
  ipe mark tiny/.style={ipe mark scale=1.1},
  ipe mark normal,
  /pgf/arrow keys/.cd,
  ipe arrow normal/.style={scale=7},
  ipe arrow large/.style={scale=10},
  ipe arrow small/.style={scale=5},
  ipe arrow tiny/.style={scale=3},
  ipe arrow normal,
  /tikz/.cd,
  ipe arrows, 
  <->/.tip = ipe normal,
  ipe dash normal/.style={dash pattern=},
  ipe dash dotted/.style={dash pattern=on 1bp off 3bp},
  ipe dash dashed/.style={dash pattern=on 4bp off 4bp},
  ipe dash dash dotted/.style={dash pattern=on 4bp off 2bp on 1bp off 2bp},
  ipe dash dash dot dotted/.style={dash pattern=on 4bp off 2bp on 1bp off 2bp on 1bp off 2bp},
  ipe dash normal,
  ipe node/.append style={font=\normalsize},
  ipe stretch normal/.style={ipe node stretch=1},
  ipe stretch normal,
  ipe opacity 10/.style={opacity=0.1},
  ipe opacity 30/.style={opacity=0.3},
  ipe opacity 50/.style={opacity=0.5},
  ipe opacity 75/.style={opacity=0.75},
  ipe opacity opaque/.style={opacity=1},
  ipe opacity opaque,
]
\definecolor{red}{rgb}{1,0,0}
\definecolor{blue}{rgb}{0,0,1}
\definecolor{green}{rgb}{0,1,0}
\definecolor{yellow}{rgb}{1,1,0}
\definecolor{orange}{rgb}{1,0.647,0}
\definecolor{gold}{rgb}{1,0.843,0}
\definecolor{purple}{rgb}{0.627,0.125,0.941}
\definecolor{gray}{rgb}{0.745,0.745,0.745}
\definecolor{brown}{rgb}{0.647,0.165,0.165}
\definecolor{navy}{rgb}{0,0,0.502}
\definecolor{pink}{rgb}{1,0.753,0.796}
\definecolor{seagreen}{rgb}{0.18,0.545,0.341}
\definecolor{turquoise}{rgb}{0.251,0.878,0.816}
\definecolor{violet}{rgb}{0.933,0.51,0.933}
\definecolor{darkblue}{rgb}{0,0,0.545}
\definecolor{darkcyan}{rgb}{0,0.545,0.545}
\definecolor{darkgray}{rgb}{0.663,0.663,0.663}
\definecolor{darkgreen}{rgb}{0,0.392,0}
\definecolor{darkmagenta}{rgb}{0.545,0,0.545}
\definecolor{darkorange}{rgb}{1,0.549,0}
\definecolor{darkred}{rgb}{0.545,0,0}
\definecolor{lightblue}{rgb}{0.678,0.847,0.902}
\definecolor{lightcyan}{rgb}{0.878,1,1}
\definecolor{lightgray}{rgb}{0.827,0.827,0.827}
\definecolor{lightgreen}{rgb}{0.565,0.933,0.565}
\definecolor{lightyellow}{rgb}{1,1,0.878}
\definecolor{black}{rgb}{0,0,0}
\definecolor{white}{rgb}{1,1,1}
\begin{tikzpicture}[ipe stylesheet]
  \node[ipe node]
     at (64, 704) {SETH};
  \node[ipe node]
     at (64, 640) {k-SUM};
  \node[ipe node]
     at (64, 576) {k-XOR};
  \node[ipe node]
     at (128, 704) {k-OV};
  \node[ipe node]
     at (64, 768) {SCC};
  \node[ipe node]
     at (192, 704) {LD-C-SETH};
  \node[ipe node]
     at (284, 704) {pw-SETH};
  \node[ipe node]
     at (376, 704) {tw-SETH};
  \draw[ipe pen fat, ->]
    (96, 708)
     -- (120, 708);
  \draw[ipe pen fat, ->]
    (160, 708)
     -- (184, 708);
  \draw[shift={(252, 708)}, xscale=1.75, ipe pen fat, ->]
    (0, 0)
     -- (16, 0);
  \node[ipe node]
     at (128, 640) {4-SUM};
  \node[ipe node]
     at (128, 576) {4-XOR};
  \draw[ipe pen fat, ->]
    (96, 644)
     -- (124, 644);
  \draw[ipe pen fat, ->]
    (96, 580)
     -- (124, 580);
  \draw[shift={(344, 708)}, xscale=2, ipe pen fat, ->]
    (0, 0)
     -- (12, 0);
  \draw[shift={(88, 772)}, yscale=0.8571, ipe pen fat, ->]
    (0, 0)
     .. controls (76, 0) and (176, -24) .. (200, -56);
  \draw[shift={(280, 722.901)}, xscale=1.3333, yscale=1.0302, ipe pen fat]
    (0, 0) rectangle (48, -96);
  \node[ipe node]
     at (284, 684) {Ind Set/pw};
  \node[ipe node]
     at (284, 668) {Dom Set/pw};
  \node[ipe node]
     at (284, 652) {Coloring/pw};
  \node[ipe node, font=\huge]
     at (284, 636) {\ldots};
  \draw[shift={(367.999, 722.901)}, xscale=1.3333, yscale=1.0302, ipe pen fat]
    (0, 0) rectangle (48, -96);
  \node[ipe node]
     at (372, 684) {Ind Set/tw};
  \node[ipe node]
     at (372, 668) {Max Cut/tw};
  \node[ipe node]
     at (372, 652) {Coloring/tw};
  \node[ipe node, font=\huge]
     at (372, 636) {\ldots};
  \draw[red, ipe pen ultrafat, ->]
    (96, 644)
     .. controls (132, 672) and (172, 672) .. (197.3333, 660.6667)
     .. controls (222.6667, 649.3333) and (233.3333, 626.6667) .. (256.6667, 615.3333)
     .. controls (280, 604) and (316, 604) .. (376, 624);
  \draw[red, ipe pen ultrafat, ->]
    (96, 580)
     .. controls (92, 604) and (74, 612) .. (62.3333, 620)
     .. controls (50.6667, 628) and (45.3333, 636) .. (52.6667, 646)
     .. controls (60, 656) and (80, 668) .. (160, 668);
  \node[ipe node]
     at (96, 712) {\cite{Williams05}};
  \node[ipe node]
     at (224, 756) {\cite{Lampis25}};
  \node[ipe node]
     at (252, 712) {\cite{abs-2407-09683}};
\end{tikzpicture}\caption{Summary of the
relations between some notable fine-grained hypotheses with the new
implications of this paper marked in red. In particular, the \ksum\ and \kxor\
Hypotheses now imply tight lower bounds for several standard problems
parameterized by treewidth. We give definitions of all hypotheses and
references to all other relations in Section~\ref{sec:prelim}, but mention
briefly that \scc\ refers to the \textsc{Set Cover Conjecture} and
\textsc{LD-C-SETH} refers to the \seth\ for $n$-input boolean circuits of size
$s$ and depth $\eps n$.}\label{fig:summary}

\end{figure}

\subparagraph*{Applications}

The immediate application of our results is that we are able to revisit several
lower bounds on the complexity of standard problems parameterized by treewidth.
These lower bounds were initially known under the \seth\ \cite{LokshtanovMS18}
and were more recently shown under the \ppseth\ \cite{Lampis25}. Our main
reduction allows us to also obtain them as consequences of the \ksum\ and
\kxor\ Hypotheses, hence arguably increasing our confidence in their validity.

To obtain these results we rely in part on the work of Iwata and Yoshida
\cite{IwataY15} who showed that falsifying the \twseth\ is equivalent to
improving the state-of-the-art DP algorithm for several other problems
parameterized by treewidth. By using their reductions, and also adding some new
results building on the framework of \cite{Lampis25} we obtain the following:

\begin{restatable}{theorem}{thmapps}\label{thm:applications} If the \ksum\
hypothesis is true, or the \kxor\ hypothesis is true, then all the following
statements are true, for all $\eps>0$:

\begin{enumerate}

\item There is no algorithm solving \textsc{Independent Set} in time
$(2-\eps)^{\tw}n^{O(1)}$.

\item There is no algorithm solving \textsc{Independent Set} in time
$(2-\eps)^{\cw}n^{O(1)}$, where $\cw$ is the width of a given clique-width
expression.

\item There is no algorithm solving \textsc{Max Cut} in time
$(2-\eps)^{\tw}n^{O(1)}$.

\item For all $k\ge 3$, there is no algorithm solving $k$-\textsc{Coloring} in
time $(k-\eps)^{\tw}n^{O(1)}$.

\end{enumerate}

\end{restatable}

Following standard conventions, the lower bounds above hold assuming that we
are given in the input an $n$-vertex graph along with a tree decomposition or
clique-width expression of size polynomial in $n$ and width $\tw$ or $\cw$
respectively. In other words, the lower bounds are not concerned with the
complexity of computing the decomposition, but with the complexity of
exploiting it to solve the corresponding problem.

The first two items of Theorem~\ref{thm:applications} follow immediately from
our main result and the results of \cite{IwataY15}, where it was shown that
achieving the stated running time for \textsc{Independent Set} for parameters
treewidth and clique-width is equivalent to falsifying the \twseth. The third
item also follows from our result and a reduction from the result of
\cite{IwataY15} for \textsc{Max-2-SAT}.  The last item is a reduction we give
here, which, however, relies on standard techniques and straightforwardly
generalizes the reduction from \ppseth\ to $k$-\textsc{Coloring} given in
\cite{Lampis25}.

These results should be interpreted as a \emph{small representative sample} of
what is possible. Indeed, one of the main thrusts of the work of Lampis
\cite{Lampis25} was to show how \seth-based lower bounds can, with a minimal
amount of effort, be transformed to pathwidth-preserving reductions from
\textsc{SAT}. This allows us to lift \seth-based lower bounds to \ppseth-based
ones and \cite{Lampis25} verified that this can be done for a wide variety of
problems (see \cite{HartmannM25,LampisL26,abs-2502-14161} for follow-up works). It is not
hard to verify by reviewing the literature than in the vast majority of cases,
known pathwidth-preserving reductions are in fact also treewidth-preserving, so
we expect that, as a rule, it should not be a problem to lift lower bounds for
problems parameterized by treewidth to \twseth-based bounds, and hence obtain
\ksum-hardness for many other problems.  However, to keep the presentation
manageable, we only focus on the problems mentioned above.

\subsection{Comparison with previous results} 

Let us give some context about \textsc{SAT}-based and \textsc{Sum}-based
hypotheses and their relations, explaining why we have selected to focus on
these particular hypotheses.

We attempt to prove an implication from the \ksum\ hypothesis towards a weaker
variant of the \seth\ and we have already explained why this weaker variant is
still useful for our application (lower bounds for treewidth-based DP).  Let us
also explain why we start from the hypothesis that \ksum, for large enough $k$,
rather than $4$-\textsc{SUM}, is hard.  For this, it is instructive to revisit
a recent result of \cite{Lampis25}, which established that the $k$-\textsc{OV}
Hypothesis implies the \ppseth.  Recall that this hypothesis states that
$k$-\textsc{Orthogonal Vectors} cannot be solved faster than roughly $n^k$
time\footnote{See the next section for a more careful definition.}.  The result
of \cite{Lampis25} is then obtained by constructing a formula with pathwidth
$k\log n$. Then, a purported algorithm solving \textsc{SAT} in $(2-\eps)^{\pw}$
would run in time $n^{(1-\delta)k}$ for the instances constructed by the
reduction. There is, however, one complication with this strategy, which is
that the purported \textsc{SAT} algorithm is allowed to have an arbitrary
polynomial dependence on the formula size (which logically should be at least
linear in the original input size $n$). In order to render this extra
dependence irrelevant, we then need to allow $k$ to become sufficiently large,
and this is the reason why the results of \cite{Lampis25} only establish that
the $k$-\textsc{OV} Hypothesis implies the \ppseth, rather than that the
$2$-\textsc{OV} Hypothesis does (which would have been stronger).  In a similar
fashion, we are forced to rely on the \ksum\ Hypothesis. 

\subparagraph*{\ksum\ is not a brute-force problem} We have therefore explained
why it is natural to start from \ksum\ and \kxor, and why we need to reduce to
a weaker variant of the \seth. Let us now also explain how the problem we are
solving is different from (and significantly harder than) previous work, and in
particular, from the proof that the $k$-\textsc{OV} Hypothesis implies the
\ppseth\ we mentioned above. The crucial difference is that whereas for the
$k$-\textsc{OV} problem the (conjectured) best algorithm is essentially a
brute-force algorithm that tries all solutions, this is \emph{not} the case for
\ksum\ or \kxor. Informally, it is easy to encode the workings of a brute-force
algorithm with a CNF formula and then show that if we could decide the
satisfiability of the formula faster than expected, we would be able to predict
the answer given by the algorithm. This becomes much more complicated when the
algorithm does something more clever than non-deterministically picking a
solution and then verifying it. 

The fact that the ``obvious'' brute-force algorithm for $3$-\textsc{SUM} is not
optimal is a known obstacle to reductions from this problem. In fact, it was
already observed by P{\u{a}}tra{\c{s}}cu \cite{Patrascu10} who was attempting
to reduce $3$-\textsc{SUM} to various other problems.  P{\u{a}}tra{\c{s}}cu's
solution was to reduce $3$-\textsc{SUM} to a $3$-\textsc{Sum-Convolution}
problem, where the obvious algorithm was hypothesized to be optimal, and indeed
this problem proved to be extremely useful in reductions establishing
$3$-\textsc{SUM} hardness. Unfortunately, this solution is not suitable for us
because the key property of the $3$-\textsc{Sum-Convolution} problem is that
the index of the last integer to be selected is fully specified from the
indices of the others.  For $k=3$ this makes brute-force optimal, but for
larger $k$ this only decreases the obvious search space from $n^k$ to
$n^{k-1}$, whereas we need to bring this down to $n^{k/2}$. For this reason we
need to rely on completely different ideas to formulate our reduction from
\ksum\ and \kxor\ to \textsc{SAT}.

To summarize, we have explained why we are reducing from \ksum\ (rather than
$4$-\textsc{SUM}) to \textsc{SAT}. We then explained that the main challenge in
making the reduction efficient enough is that the (conjectured) optimal
algorithm for \ksum\ is much more clever than a brute-force algorithm that
guesses a solution and then verifies it. We will therefore exploit the full
power afforded to us by an algorithm solving \textsc{SAT} given a graph
decomposition and in particular we will crucially rely on the extra power of
tree over path decompositions.

\subsection{Techniques and proof overview}

We will formulate a reduction from \ksum\ or \kxor\ to \textsc{SAT} instances
of treewidth essentially $\frac{k}{2}\log n$. As a result, if there is an
algorithm solving \textsc{SAT} in time $(2-\eps)^{\tw}|\phi|^{O(1)}$, we will
obtain an algorithm for \ksum\ and \kxor\ with complexity
$n^{(1-\delta)\frac{k}{2}}$. Let us describe step-by-step our approach, the
obstacles we face, and our solutions.  Along the way, we will clarify the
subtle but intriguing point of why the consequence we obtain is the \twseth,
rather than the \ppseth.

The first step is to recall that \textsc{SAT} formulas of pathwidth $\pw$ have
a tight correspondence with non-deterministic algorithms using space
$\pw+O(\log n)$, as shown by Iwata and Yoshida \cite{IwataY15}. Intuitively,
this is established via a Cook-Levin type construction: Given such a machine,
we build a \textsc{SAT} formula and a path decomposition such that each bag
corresponds to a snapshot of the memory of the non-deterministic machine and we
only need constraints to ensure that consecutive snapshots correspond to a
valid transition. Hence, if the machine uses $\pw$ bits of memory, each bag
would have size essentially $\pw$.

This seems like an encouraging start, as \ksum\ is clearly in NL. However, the
obvious non-deterministic logarithmic space algorithm for \ksum\ is to guess
the $k$ indices of the solution and then verify it. This algorithm uses space
$k\log n$, therefore corresponds to a formula of pathwidth $k\log n$, rather
than $\frac{k}{2}\log n$. This is exactly the point where the obstacle we
discussed above (that the best algorithm for \ksum\ is not brute-force search)
manifests itself.

We are therefore motivated to come up with an alternative non-deterministic
algorithm for \ksum, using space only roughly $\frac{k}{2}\log n$. This space
is clearly not enough to store the indices of all selected elements. However,
an alternative approach to solve \ksum\ (and \textsc{Subset Sum})
non-deterministically is the following: Consider each element one by one and
non-deterministically guess if this item must be selected. Along the way,
maintain a running sum, that is, the sum of all selected items. This algorithm
is a non-deterministic version of the standard dynamic programming algorithm.
If the maximum absolute value in the input is $W$, it uses space $\log W+O(\log
n)$, so it would suffice if $W=O(n^{k/2})$.  Unfortunately, \ksum\ under this
restriction is not known to be equivalent to the general case.  Instead, what
is known is that by applying appropriate hash functions we can edit an
arbitrary \ksum\ instance to ensure that $W=O(n^k)$ (see Lemma 2.2 or
\cite{DalirrooyfardLS25}).  This seems like a dead-end because it gives the
same space complexity as the trivial algorithm.  Nevertheless, we will push
this idea by attempting to further decrease $W$.

Of course, it seems hopeless (to us)  to attempt to show that \ksum\ remains
hard on instances with $W=O( n^{k/2})$.  To see why, recall that the way that
$W$ can be brought down to roughly $n^k$ is to apply to the input a hash
function that is (almost-)linear, meaning it (almost) satisfies
$h(x+y)=h(x)+h(y)$ for all inputs $x,y$. Applying this function means that all
$k$-tuples which sum to $0$ become hashed $k$-tuples that sum to $0$; while any
other $k$-tuple will map to $0$ with probability $O(\frac{1}{W})$.  Since there
are $n^k$ possible $k$-tuples, if we set $W$ a bit larger than $n^k$ we can
claim by union bound that false positives are unlikely. However, if $W$ is
smaller than $n^k$ the probability of false positives rises quite quickly.

In order to deal with this obstacle we return to the root of our problems,
which is the meet in the middle algorithm. Observe that \ksum\ and \kxor\ can
be seen as instances of \textsc{List Disjointness}: if $L$ is the set of the
(at most $n^{k/2}$) sums constructible from the first $k/2$ arrays and $R$ is
the set of sums constructible from the rest, the problem is essentially to
check if $L\cap R\neq\emptyset$. The meet in the middle algorithm relies on the
fact that \textsc{List Disjointness} for $|L|=|R|=N$ takes $\widetilde{O}(N)$
rather than $O(N^2)$ time to solve. This formulation nicely meshes with the
idea of using hash functions, because one way to test if $L\cap R\neq
\emptyset$ is to apply a hash function to the two sets and then check
disjointness in the hashed values. This seems to improve our situation because
we are mapping a collection of $|L\cup R|=O(n^{k/2})$ elements (compared to
$n^k$ $k$-tuples). However, in order to produce a collision-free hash function,
we need the range we are mapping to to be \emph{quadratic} in the number of
elements (due to the birthday paradox). So this still gives $W$ in the order of
$n^k$.

We now arrive at the key idea of our construction, which will make clear why we
need to use treewidth rather than pathwidth. We are in a situation where we
wish to map $N=2n^{k/2}$ elements (the set $L\cup R$) into a range of size
roughly $N$, while avoiding collisions, so we draw some ideas from the
classical 2-level hashing scheme of Fredman, Koml{\'{o}}s, and Szemer{\'{e}}di
\cite{FredmanKS84}. The main idea of this scheme is to use an initial
\emph{main} hash function to map the $N$ elements into $N$ buckets; then, if
the hash function is good enough, even though there will be some collisions,
all buckets will have small load; we can therefore use a secondary hash
function inside each bucket to eliminate the remaining collisions. In
particular, the secondary hash function can afford to be quadratic in the size
of each bucket and we know that in this case the secondary hash function is
collision-free with probability at least $\frac{1}{2}$.

In order to illustrate how this idea will be used in our setting consider the
following \emph{Interactive Proof} protocol for \ksum\ (formulated as
\textsc{List Disjointness}). A verifier selects a random \emph{main} hash
function $h^*$ which maps the $N=O(n^{k/2})$ elements of $L\cup R$ to a range
of size $N$. The prover is then challenged to produce a value $v_1$ such that
there exist $x\in L$ and $y\in R$ which satisfy $h^*(x)=v_1=h^*(y)$. The prover
can clearly do this if a solution exists, but he will also be able to do it if
$h^*$ has some collisions and this will happen with high probability. Suppose,
however, that $h^*$ is \emph{balanced}, in the sense that with high
probability, the maximum load $M$ of any bucket is sub-polynomial in $N$
($M=N^{o(1)}$). The verifier then selects a random secondary hash function
$h_2$ with range $M^2$, so with probability at least $\frac{1}{2}$ this
function is collision-free for the bucket $v_1$. The prover is then challenged
to produce $x\in L, y\in R$ such that $(h^*(x),h_2(x)) = (v_1,v_2) =
(h^*(y),h_2(y))$, that is, the prover is challenged to produce a collision in
$h_2$ \emph{while respecting the commitment made} for the value of $h^*$. If no
true solution exists, this will be possible with probability at most
$\frac{1}{2}$, and repeating this second step many times can make this
probability very small.

Observe that the interactive protocol we sketched above has some positive and
some negative aspects. On the positive side, the \emph{memory} the verifier
needs to use is as small as desired: we need to keep track of the prover's
initial response ($\frac{k}{2}\log n$ bits), while the responses given in each
round consist of $o(k\log n)$ additional bits (because $M^2=N^{o(1)}$), which
can be discarded after the answer is verified. On the negative side, this
protocol seems to differ significantly from a non-deterministic algorithm,
because we force the prover to commit to a value $v_1=h^*(x)=h^*(y)$ before we
make additional queries. This seems to involve a quantifier alternation:
whereas one round of interaction can be encoded by a non-deterministic machine
that guesses the prover's response, for several rounds a machine would have to
guess the prover's initial response, so that \emph{for all} subsequent requests
with distinct secondary hash functions \emph{there exists} a prover response
that convinces the verifier.

We have finally arrived at the justification for the use of treewidth. As
mentioned, there is a tight connection between the \ppseth\ and
non-deterministic log-space machines, so barring an improvement to the above
protocol that would eliminate the quantifier alternation \emph{without
increasing} the verifier's memory, it seems very hard to obtain a formula of
the desired pathwidth. Nevertheless, recent works on the classes XALP and XNLP
give us the intuition that whereas pathwidth captures non-deterministic
log-space algorithms, treewidth captures \emph{alternating} log-space
algorithms\footnote{With a logarithmic number of alternations, but this is fine
for our purposes.}, informally because Introduce nodes are non-deterministic
(we have to guess the value of a new variable), but Join nodes are
co-non-deterministic (we have to ensure the current assignment can be extended
to both children).  The key intuition is then that we can encode the above
protocol in a formula of \emph{treewidth} $(1+o(1))\frac{k}{2}\log n$.  This is
achieved by constructing a separate formula of \emph{pathwidth}
$(1+o(1))\frac{k}{2}\log n$ for each round of the interaction and then adding
constraints that ensure the prover has to remain consistent with respect to the
main hash function $h^*$. These extra constraints transform the union of the
path decompositions into a tree decomposition.

\subparagraph*{Putting it all together} Let us now sketch how the ideas
described above lead to the reduction. We use the same general strategy for
\kxor\ and \ksum, though for reasons we describe below the reduction is simpler
for \kxor. Initially, we pick a main hash function $h^*$ that is (almost)
linear and has range of the order $n^{k/2}$. Suppose that we know that $h^*$ is
balanced in the sense that, for all inputs, with high probability the maximum
load of any bucket is $M=n^{o(k)}$.  We pick many (say $10k\log n$) secondary
hash functions $h_1,\ldots,h_{10k\log n}$, with range of the order $M^2$. 

Now, for every pair $(h^*,h_\ell)$, for $\ell\in[10k\log n]$, we construct a
\textsc{SAT} instance $\phi_\ell$ of pathwidth $\frac{k}{2}\log n+O(\log M) =
(1+o(1))\frac{k}{2}\log n$. The formula $\phi_\ell$ encodes the execution of
the non-deterministic \textsc{Subset Sum} algorithm which uses a counter to
keep track of the current sum, applied to the input instance after we apply to
each element the pair of hash functions $(h^*,h_\ell)$. The bags of the
decomposition of $\phi_\ell$ are now supposed to contain the contents of the
memory of the algorithm after every step. We focus on the ``mid-point'' bag
that represents the sum calculated right after the first $k/2$ arrays have been
processed. For each $\ell\in[10k\log n-1]$ we add constraints ensuring the
formulas $\phi_\ell$ and $\phi_{\ell+1}$ agree on the value of $h^*$ at the
mid-point of the execution of the algorithm. Hence, we obtain a formula of
treewidth $(1+o(1))\frac{k}{2}\log n$ which simulates the interactive protocol
we discussed above. Namely, in order to satisfy the formula one would first
have to commit to a value of the main hash function at the mid-point (that is,
to a value $v_1=h^*(x)=h^*(y)$) and then produce a collision for all secondary
hash functions while respecting the commitment. This intuition drives the proof
of correctness.

Finally, let us explain how the reductions for \kxor\ and \ksum\ differ. When
dealing with \kxor\ we use linear transformations (we multiply each input
vector with a random boolean matrix). This makes it easy to concatenate two
hash functions $(h^*,h_{\ell})$. Crucially, linear hash functions in this
context have been shown to achieve maximum load comparable to a fully random
function, that is, \emph{logarithmic} in the number of elements
\cite{AlonDMPT99,JaberKZ25}. Since for our purposes any sub-polynomial bound on
the maximum load works, we can put these ingredients together and obtain a
reduction along the lines sketched above.

In contrast to this, for \ksum\ we have to deal with several additional
problems. First, concatenating hash functions means we are dealing with a
\textsc{Vector Subset Sum} problem, but this is easy to deal with. Second, the
hash functions which are available in the context of integers (we use a
function due to Dietzfelbinger \cite{Dietzfelbinger96}) are generally
almost-linear, in the sense that due to rounding errors, $h(x+y)$ is only
guaranteed to be close to $h(x)+h(y)$, rather than equal. This is a common
difficulty in this context and is also not too hard to deal with. The main
problem is that for this class of almost-linear hash functions over the
integers it is currently an open problem if we can guarantee a maximum load
better than $M=N^{1/3}$ when mapping $N$ elements into a range of size $N$
\cite{Knudsen19}.  This forces us to use a function that is made up of a
concatenation of many smaller independent hash functions. Then, using the fact
that the base function is pair-wise independent
\cite{Dietzfelbinger96,Dietzfelbinger18} we are able to show that the maximum
load can be upper-bounded sufficiently well for our purposes.

Because of these added difficulties for \ksum, we begin our presentation in
Section~\ref{sec:kxor} with the reduction for \kxor. Then, in
Section~\ref{sec:ksum} we mainly focus on overcoming the additional obstacles
for \ksum\ and assume the reader is already familiar with the general strategy
of the construction. We therefore suggest to the interested reader to begin
with Section~\ref{sec:kxor}. Applications of our main result are given in
Section~\ref{sec:applications}.

\subsection{Open problems}

In general, we believe it would be fruitful to continue investigating the
connections between different flavors of fine-grained complexity hypotheses, as
this can both increase our confidence in the validity of our lower bounds and
provide valuable insight. Beyond this general direction, let us mention some
concrete questions directly relevant to this work:

\subparagraph*{Pathwidth} The most natural problem we leave open is the
following: does the \ksum\ (or \kxor) Hypothesis imply the \ppseth? We have
already explained the reasons why our current proof needs to rely on treewidth,
rather than pathwidth. Furthermore, by the characterization of Iwata and
Yoshida \cite{IwataY15}, if it were possible to obtain a reduction from \ksum\
to \textsc{SAT} with pathwidth $\frac{k}{2}\log n$ (which seems necessary),
then we would have a non-deterministic algorithm solving \ksum\ in the same
amount of space. The ``easy'' algorithmic ideas one would try clearly fall
short of this, so it merits further investigation whether this can be done, or
whether some circumstantial evidence against such an algorithm can be found.

\subparagraph*{Derandomization} Another natural problem we have not addressed
is the question of whether our reduction can be made deterministic. Even though
reductions from $3$-\textsc{SUM} and \ksum\ are often randomized, it has
recently become a topic of interest to investigate whether such reductions can
be derandomized \cite{ChanH20,FischerK024}. Can these ideas be applied to our
reduction?

\subparagraph*{Clique} Moving away from \ksum, it would be very interesting to
make a connection between \textsc{SAT}-based hypotheses, where it is
conjectured that we have to try all possibilities in some sense, and other
problems where more clever algorithms exist. A very notable case is detecting
$k$-cliques in a graph, which can be done in time $n^{\frac{\omega k}{3}}$,
where $\omega$ is the matrix multiplication constant.  Our concrete question is
then the following: if the \twseth\ is false, does this imply an algorithm
which can detect if a graph has a $k$-clique in time $n^{\frac{\omega
k}{3}-\eps}$ or even $n^{\frac{2k}{3}-\eps}$? In the same way that to obtain
the results of this paper we had to encode the main idea of the meet in the
middle algorithm for \ksum, it appears likely that to obtain this implication
we would need a reduction that needs some elements of fast matrix
multiplication algorithms, or at least, efficient algorithms which can verify
the correctness of matrix multiplication.

We stress here that in the question formulated above we are discussing
$k$-clique detection, because this is a problem for which a non-trivial
algorithm is known and is conjectured to be optimal. For the (harder) problems
of detecting a $k$-clique in a hypergraph or a $k$-clique of minimum
(edge-)weight, it is known that if these problems are as hard as conjectured
($n^k$), then the $k$-\textsc{OV} Hypothesis is true \cite{AbboudBDN18},
therefore the \ppseth\ is true. Even without using the results of
\cite{AbboudBDN18} it is straightforward to show that these problems can be
encoded by a CNF formula of pathwidth $k \log n$, giving the implication to the
\ppseth. What makes the $k$-clique detection question interesting, then, is
that the conjectured best bound does not correspond to brute-force search.

\subsection{Other related work}

As mentioned, fine-grained complexity theory is a very vibrant area and
numerous complexity hypotheses can be found in the literature. In this paper we
focus on variants of the Strong Exponential Time Hypothesis (\seth),
formulated by Impagliazzo and Paturi \cite{ImpagliazzoP01}, which states that
for all $\eps$ there exists $k$ such that $k$-\textsc{SAT} on inputs with $n$
variables cannot be solved in time $(2-\eps)^n$. The $3$-\textsc{SUM}
Hypothesis goes back to \cite{GajentaanO95} and has been used as the basis for
numerous fine-grained lower bounds. As mentioned, a recent breakthrough result
established that this hypothesis is in fact false \cite{alman2026}, without
however falsifying the $4$-\textsc{SUM} hypothesis or the \ksum\ hypothesis
that we use in this paper.  It would be impractical to review here the numerous
other fine-grained hypotheses that can be found in the literature.  We do
mention, however, two other hypotheses which are close to our topic.  The
\textsc{Set Cover Conjecture} \cite{CyganDLMNOPSW16} states that \textsc{Set
Cover} on instances with $m$ sets and $n$ elements cannot be solved in time
$(2-\eps)^nm^{O(1)}$; and the \textsc{Max}-$3$-\textsc{SAT} Hypothesis, which
states that \textsc{Max}-$3$-\textsc{SAT} cannot be solved in time $(2-\eps)^n$
\cite{EsmerFMR24}. It is an open problem whether either one of these hypotheses
implies or is implied by the \seth.

Establishing connections between the different branches of hypotheses of
fine-grained complexity theory is a topic that has already attracted attention.
A previous work that is very close in spirit to what we do here is that of
Abboud, Bringmann, Dell, and Nederlof \cite{AbboudBDN18}, who showed that if it
is hard to find a minimum weight $k$-clique in time $n^{(1-\eps)k}$ or if it is
hard to detect a $k$-clique in a hypergraph in the same amount of time, then
the $k$-\textsc{OV} problem also requires time $n^{(1-\eps)k}$. 

More broadly, motivated by the somewhat speculative nature of many assumptions,
much work has been devoted to reproving lower bounds known under some
hypothesis, while using a weaker or orthogonal hypothesis. For instance,
Abboud, Bringmann, Hermelin, and Shabtay \cite{AbboudBHS22} show that the
optimality of the textbook DP algorithm for \textsc{Subset Sum}, which was
previously known under the \scc, is also implied by the \seth. More broadly,
Abboud, Vassilevska Williams, and Yu \cite{AbboudWY18} started an effort to
define a problem whose hardness is implied by all three major hypotheses, with
the hope that this problem would serve as a more solid starting point for
reductions. In a similar vein, Abboud, Hansen, Vassilevska Williams, and
Williams \cite{AbboudHWW16} advocated for the use of the \textsc{NC-SETH} (the
version of the \seth\ for circuits of poly-logarithmic depth) as a more solid
starting point for reductions than the \seth.

We also note that the validity of hypotheses used in fine-grained complexity
should not in general be taken for granted. In addition to the mentioned recent
breakthrough falsifying the $3$-\textsc{SUM} hypothesis \cite{alman2026},
recent work by Bj{\"{o}}rklund and Kaski \cite{BjorklundK24} followed up by
Pratt \cite{Pratt24} established that the \scc\ and the Tensor Rank Conjecture
(a conjecture from the realm of matrix multiplication algorithms) cannot both
be true. Also, Williams has disproved a plausible-sounding strengthening of the
non-deterministic \seth\ \cite{Williams16}.

Several recent works have focused on proving that the \seth\ \emph{does not}
imply various lower bounds, often using the non-deterministic \seth\ as a
hypothesis. In particular, it was established in \cite{CarmosinoGIMPS16} that
if one could show that the \seth\ implies the $3$-\textsc{SUM} hypothesis (or
the \textsc{APSP} hypothesis), then the \emph{non-deterministic} version of the
\seth\ would be false\footnote{Roughly speaking, the non-deterministic version
of the \seth\ states that $k$-\textsc{SAT} cannot be solved in time
$(2-\eps)^n$ even by co-non-deterministic algorithms.}. In a similar spirit,
Aggarwal and Kumar \cite{Aggarwal023} show that hardness for some cases of the
\textsc{Closest Vector Problem} cannot be based on the \seth, Trabelsi shows a
similar result for \textsc{All-Pairs Max Flow} \cite{Trabelsi25}, and Li shows
a similar result for some cases of approximating the diameter \cite{Li21a}.
Furthermore, Belova et al.  \cite{BelovaKMRRS24}, building on previous work
\cite{BelovaGKMS23}, show that proving \seth-based hardness for various
problems in P will be challenging. In particular, a result of
\cite{BelovaKMRRS24} that is very relevant to us is that proving a \seth-based
lower bound of $n^{1+\eps}$ for \ksum\ would require disproving a stronger
version of the \seth, which would in turn imply some circuit lower bounds.
Note, however, that these results discuss the difficulty of reducing
\textsc{SAT} to \ksum, while our main result is a reduction in the converse
direction. A recent work that does go in the same direction as we do in this
paper is that of Belova, Chukhin, Kulikov, and Mihajlin \cite{BelovaCKM24}, who
reduce \textsc{Subset Sum} to the \textsc{Orthogonal Vectors} problem. However,
their motivation is to obtain a \textsc{Subset Sum} algorithm using less space,
rather than to show hardness for \textsc{Orthogonal Vectors}.

The intuition that \textsc{SAT} parameterized by pathwidth encodes
non-deterministic logarithmic space algorithms comes from \cite{IwataY15},
while the intuition that the difference between treewidth and pathwidth in this
setting is that treewidth simulates alternations comes from works on the
classes XNLP and XALP \cite{BodlaenderGJJL22,BodlaenderGJPP22,BodlaenderGNS21}.
This intuition goes back to the work of Allender, Chen, Lou, Papakonstantinou,
and Tang \cite{AllenderCLPT14} who showed that whereas \textsc{SAT} on
instances of pathwidth $O(\log n)$ is complete for NL, \textsc{SAT} on
instances of treewidth $O(\log n)$ is complete for SAC$^1$, the class of
problems decidable by logarithmic depth boolean circuits of semi-unbounded
fan-in. This is a somewhat weaker class than AC$^1$ (circuits of unbounded
fan-in), which in turn corresponds exactly to logarithmic space alternating
non-deterministic machines which perform $O(\log n)$ alternations.

\section{Preliminaries}\label{sec:prelim}

As mentioned, we expect the reader to be familiar with the basics of
parameterized complexity, such as DP algorithms parameterized by treewidth, and
refer to the standard textbook for the relevant definitions
\cite{CyganFKLMPPS15}.  For a CNF formula $\phi$ we define the primal graph of
$\phi$ to be the graph obtained if we construct a vertex for each variable of
$\phi$ and make two vertices adjacent if the corresponding variables appear in
a clause together.  When we refer to the treewidth or pathwidth of a formula
$\phi$ we mean the corresponding width of the primal graph.

Regarding the hypotheses of Figure~\ref{fig:summary} we have the following
definitions:

\begin{enumerate}

\item \seth: For all $\eps>0$ there exists $k$ such that $k$-\textsc{SAT}
cannot be solved in time $O((2-\eps)^n)$ on instances with $n$ variables.

\item \ppseth: For all $\eps>0$ there exists no algorithm solving \textsc{SAT}
in time $(2-\eps)^\pw|\phi|^{O(1)}$, where $\pw$ is the width of a given path
decomposition of the input formula $\phi$ that has size at most
$|\phi|^{O(1)}$.

\item \twseth: The same as the previous statement, but given a tree
decomposition instead of a path decomposition.

\item \textsc{LD-C-SETH}: For all $\eps>0$, there exists no algorithm which
takes as input a bounded fan-in boolean circuit with $n$ inputs and
depth at most $\eps n$ and decides if the circuit is satisfiable in time
$O((2-\eps)^n)$.

\item $k$-\textsc{OV} Hypothesis: For all $k\ge 2, \eps>0$, there exists no
algorithm which takes as input $k$ arrays, each containing $n$ boolean vectors
of length $d$, and decides if there is a way to select one vector from each
array so that the selected vectors have product $0$ in time
$n^{k-\eps}d^{O(1)}$.

\item \scc: For all $\eps>0$ there exists $k$ such that no algorithm can solve
\textsc{Set Cover} on inputs with $n$ elements and $m$ sets of size at most $k$
in time $(2-\eps)^nm^{O(1)}$.

\end{enumerate}

We have already defined \ksum\ and \kxor\ in the introduction. Regarding the
relations between hypotheses, it was shown by Williams that the \seth\ implies
the $k$-\textsc{OV} Hypothesis \cite{Williams05}, while it is easy to reduce
the $k$-\textsc{OV} problem to the satisfiability of a boolean circuit with
$k\log n$ inputs and small depth. The \textsc{LD-C-SETH} is a hypothesis
posited in \cite{abs-2407-09683}, where it was shown that it implies (a
stronger form of) the \ppseth. Notice that the \textsc{NC-SETH}, mentioned in
\cite{AbboudHWW16}, trivially implies the \textsc{LD-C-SETH}. The fact that the
\scc\ implies the \ppseth\ was shown in \cite{Lampis25}.

In general we try to use standard graph-theoretic notation. When we write
$[i,j]$ for $i,j$ positive integers, we mean the set of all integers between
$i$ and $j$ inclusive, that is, $[i,j]=\{i,i+1,\ldots,j\}$. When we write $[i]$
we mean $[1,i]$. A function $f(x)$ is affine if $f(x)=ax+b$, for $a,b$
constants, where we call $b$ the offset. A function is linear if it is affine
and has offset $0$.

We do not pay much attention to the specific machine model, since we ignore
poly-logarithmic factors, but the reader may assume we are using the RAM model.
We also in general allow randomized algorithms and our main result will be a
reduction that succeeds with high probability and one-sided error. For \ksum\
and \kxor\ one can assume that the given integers/vectors have length bounded
by $O(\log n)$ \cite{DalirrooyfardLS25}, but one point of focus will be to
determine the constant hidden in the O-notation and how it depends on $k$.

\section{\kxor\ hardness implies \twseth}\label{sec:kxor}

Our goal in this section is to show the following theorem:

\begin{restatable}[\kxor\ hypothesis $\Rightarrow$
\twseth]{theorem}{kxorthm}\label{thm:kxor} Suppose that there exists an
$\eps>0$ and an algorithm which, given a CNF formula $\phi$ and a tree
decomposition of its primal graph of width $\tw$ and size polynomial in
$|\phi|$, decides if $\phi$ is satisfiable in time $(2-\eps)^\tw|\phi|^{O(1)}$.
Then, there exist $\delta>0, k_0>0$ and a randomized algorithm which for all
$k>k_0$ solves \kxor\ on instances with $k$ arrays each containing $n$ vectors
in time $n^{(1-\delta)\frac{k}{2}}$ with one-sided error and success
probability $1-o(1)$.  In other words, if the \twseth\ is false, then the
randomized \kxor\ hypothesis is false.\end{restatable}

\subparagraph*{High-level technique summary} Following the sketch given in the
introduction, we perform a reduction based on the following steps, transforming
a given \kxor\ instance into a \textsc{SAT} instance with an appropriate tree
decomposition:

\begin{enumerate}

\item We select a random linear function that cuts down the length of all
vectors of the initial instance to $\frac{k}{2}\log n$. Call this the main hash
function $h^*$.

\item Using the reduction from \textsc{Subset Sum} to \textsc{SAT} from
\cite{Lampis25} we can reduce the hashed instance to a \textsc{SAT} instance
with pathwidth $\frac{k}{2}\log n$. From the satisfying assignments to this
instance we can extract a value $y$ such that $y=h^*(v_1)=h^*(v_2)$, where
$v_1,v_2$ are vectors constructible from the first and last $k/2$ arrays
respectively. Unfortunately, $h^*$ could have many collisions, so this step
could produce many false positives (that is, pairs $v_1,v_2$ with $v_1\neq v_2$
but $h^*(v_1)=h^*(v_2)$).

\item We observe that the main hash function $h^*$ we used has a good balancing
property, as shown by the work of Alon et al.~on random linear functions
\cite{AlonDMPT99}, namely, the maximum number of collisions in each target
value is not far from what we would expect from a completely random function.

\item As a result, we can apply a secondary hash function, whose range is
quadratic in the expected bucket size (which is poly-logarithmic, rather than
polynomial, in $n$), making the inefficiency of this second function
insignificant.  The secondary hash function is collision-free with constant
probability.

\end{enumerate}

Finally, there is one last step missing to complete this construction. Observe
that, for any \emph{fixed} value $y=h^*(v_1)=h^*(v_2)$, as described in step 2
above, extending $h^*$ with a secondary hash function (of much smaller range)
will eliminate collisions with constant probability and hence make the
reduction correct. The problem is that to get the reduction to work as
described above we would need to essentially exchange the order of quantifiers
in the above statement\footnote{To see what we mean, read the statement of the
previous paragraph as ``for all buckets $y$, most secondary hash functions make
$y$ collision-free'' and compare with the (stronger) statement ``most secondary
hash functions make all buckets $y$ collision-free''.}, which does not seem
possible. Indeed, this is the main obstacle to proving that the \kxor\
hypothesis implies the \ppseth\ (rather than the \twseth).

To deal with this problem we use the implicit quantifier alternation that
becomes available to us when we parameterize \textsc{SAT} by treewidth rather
than pathwidth. We produce many (say $\Omega(k\log n)$) copies of the
construction sketched above, all using the same main hash function, but each
using a different secondary hash function. We then connect these copies in a
way that forces a satisfying assignment to commit to a value
$y=h^*(v_1)=h^*(v_2)$ throughout the instance. For any fixed $y$, the instance
is then falsely satisfiable if bucket $y$ has a collision for all $\Omega(k\log
n)$ secondary hash functions, which happens with probability $n^{-\Omega(k)}$;
hence by union bound the probability that there exists a value $y$ which
produces a false positive is very low.  The key fact here is that connecting
the instances in this way may increase the pathwidth, but we can still bound
the treewidth of the new instance by the desired bound.

In the rest of this section we first review the basic tools we will need in
Section~\ref{sec:xortools}, namely a version of the aforementioned \textsc{Subset Sum}
reduction and known results on linear hash functions; we then describe our
reduction from \kxor\ to \textsc{SAT} in Section~\ref{sec:kxorreduction}; and finally
we prove the correctness of our reduction (and hence Theorem~\ref{thm:kxor}) in
Section~\ref{sec:kxorproof}.

\subsection{Tools}\label{sec:xortools}

As mentioned we will make use of and adapt two main ingredients from previous
work. The first is a reduction from \textsc{Subset Sum} to \textsc{SAT} which
produces an instance whose pathwidth is $\log t$, where $t$ is the target
value. We describe a version of such a reduction, adapted from \cite{Lampis25},
below.  Our second ingredient is the fact that random linear hash functions
behave ``well'', in the sense that the maximum load is comparable to that of a
truly random hash function, as proved in \cite{AlonDMPT99}. We review these
facts in Section~\ref{sec:hashalon}.

\subsubsection{Subset (XOR)-Sum to pathwidth-SAT}\label{sec:xorsat}

We start with a lemma informally stating the following: suppose we are given a
\kxor\ instance, where each vector has $u$ bits. It is possible to encode this
instance using a CNF formula of pathwidth (almost) $u$ and produce a path
decomposition of the resulting formula that has the following special property:
the last bag of the decomposition encodes exactly the sums which are
constructible from the vectors of the instance. More precisely, the last bag of
the decomposition has $u$ (ordered) variables and an assignment to these
variables can be extended to a satisfying assignment of the whole formula if
and only if the vector corresponding to this assignment is constructible in the
\kxor\ instance by selecting one vector from each array and computing their
sum.

\begin{lemma}\label{lem:subsetsum} Suppose we are given $\ell$ arrays
$A_1,\ldots, A_\ell$, each containing a collection of boolean vectors of $u$
bits. There exists a polynomial-time algorithm which produces a formula $\psi$
of pathwidth $u+O(1)$ and a path decomposition of $\psi$ such that the last bag
of the decomposition $B_0$ contains $u$ variables $x_1,\ldots,x_u$ and such
that we have the following: for each truth assignment $\sigma$ to the variables
of $B_0$, $\sigma$ can be extended to a satisfying assignment of $\psi$ if and
only if there exist $i_1,i_2,\ldots,i_\ell$ such that $\sum_{j\in[\ell]}
A_j[i_j] = (\sigma(x_1),\sigma(x_2),\ldots,\sigma(x_u))$. \end{lemma}

\begin{proof}

We first describe the construction for $\ell=1$ and then reuse this inductively
to arrive at the general case.

Suppose $\ell=1$, so we are given a single array $A_1$ containing $n$ vectors
of $u$ bits. We construct a CNF formula containing the following variables:

\begin{enumerate}

\item The variables $x_{i,j}$, for $i\in[0,n], j\in[u]$. The intuitive meaning
of these is that, for a fixed $i$, the $u$ variables $x_{i,j}$ encode the sum
we have constructed after having considered the first $i$ elements of $A_1$.

\item The variables $s_i$, for $i\in[n]$. The intuitive meaning is that $s_i$
is set to $1$ if the $i$-th element of $A_1$ is selected.

\item The variables $y_i$, for $i\in[n]$. The intuitive meaning is that $y_i$
is set to $1$ if at least one of $s_1,s_2,\ldots,s_i$ is set to $1$.

\end{enumerate}

The idea is that the formula we are constructing will encode the workings of a
non-deterministic algorithm, which considers each element of $A_1$ in turn,
decides whether to select it (this is encoded by the $s_i$ variables) and keeps
track of the sum of the currently selected elements. The $y_i$ variables will
allow us to check that exactly one element per array is selected.

We add the following clauses:

\begin{enumerate}

\item For $i=1$, add clauses ensuring that $y_1=s_1$. For $i\in[2,n]$, add
clauses ensuring that $y_i=(y_{i-1}\lor s_i)$.

\item For $i\in[2,n]$ add the clause $(\neg s_i\lor \neg y_{i-1})$, ensuring
that we cannot select the $i$-th vector if we have already selected one among
the first $i-1$ vectors. Also, add the unit clause $y_n$, ensuring that at
least one vector is selected in total.

\item For $i\in[n]$, for $j\in[u]$, add clauses ensuring that if $s_i$ is set
to $0$, then $x_{i-1,j}=x_{i,j}$. That is, if the $i$-th vector is not
selected, then the value encoded by $x_{i,j}$ is the same as the value encoded
by $x_{i-1,j}$.

\item For $i\in[n]$, for $j\in[u]$ such that the bit in position $j$ of
$A_1[i]$ is $1$, add clauses ensuring that if $s_i$ is set to $1$ then
$x_{i,j}\neq x_{i-1,j}$.

\item For $i\in[n]$, for $j\in[u]$ such that the bit in position $j$ of
$A_1[i]$ is $0$, add clauses ensuring that if $s_i$ is set to $1$ then
$x_{i,j} = x_{i-1,j}$. The clauses of this and the previous step ensure that if
the $i$-th vector is taken, $x_{i,j}$ is calculated by adding the $j$-th bit of
$A_1[i]$ to $x_{i-1,j}$.

\end{enumerate}

\subparagraph*{Correctness:} The formula we have constructed is satisfiable.
Indeed, it is not hard to see that all satisfying assignments set exactly one
$s_i$ variable to $1$ and once we have fixed which $s_i$ is set to $1$ there is
a unique way to extend this to a satisfying assignment of the whole formula
(modulo the assignment to $x_{0,j}$ for $j\in[u]$): set all other $s_{i'}$, for
$i'\neq i$ to $0$; set $y_{i'}$ to $1$ if and only if $i'\ge i$; set
$x_{i',j}=x_{0,j}$ for $i'<i$; set $x_{i',j}=x_{0,j}+A_1[i]_j$ for $i'\ge i$,
where $A_1[i]_j$ is the $j$-th bit of $A_1[i]$. In other words, we can satisfy
the formula if and only if we select a single $i\in [n]$ and then ensure that
the vectors $x_{0,j}$ and $x_{n,j}$ differ exactly by $A_1[i]$. For the base
case it now suffices to add constraints to the formula so that $x_{0,j}=0$ for
all $j$.

\subparagraph*{Pathwidth:} We can construct a path decomposition by starting
with $n+1$ bags $B_i, i\in[0,n]$, with each bag $B_i$ containing the $u$
variables $x_{i,j}$. For each $i\in[n]$ we add to both $B_i$ and $B_{i-1}$ the
variables $s_i,y_i$. This has covered the clauses of the first two steps. To
cover the remaining clauses, for each $i\in[n]$ we insert a sequence of $2u$
new bags between $B_{i-1}$ and $B_i$. Start with $B_{i-1}$ and for each
$j\in[u]$ insert after the current bag a copy of the same bag with $x_{i,j}$
added, and then a copy of this second bag with $x_{i-1,j}$ removed, and make
this new bag the current bag. Repeat in this way, at each step exchanging
$x_{i-1,j}$ with $x_{i,j}$, until we arrive at $B_i$. This process covers all
remaining clauses. In the end, after the last bag we add a bag containing
$x_{n,j}$ for $j\in[u]$.

\subparagraph*{General construction:} Given the construction for $\ell=1$, we
can handle the general case of larger values of $\ell$ as follows: first,
construct a formula for the first ${\ell-1}$ arrays and a path decomposition
where the last bag contains $u$ variables which encode all constructible
vectors. Use the base case construction to encode the array $A_\ell$ with a CNF
formula, without adding the clauses that force $x_{0,j}$ to $0$; instead,
identify the $x_{0,j}$ variables with the variables of the last bag of the
decomposition of the formula representing the arrays $A_1,\ldots,A_{\ell-1}$.
This completes the construction, preserves pathwidth, and correctness can
easily be established by induction.  \end{proof}

Because in the previous lemma the ``special'' bag whose variables encode the
constructible sums is the last bag of the decomposition, we can glue together
two copies of this construction to encode a \kxor\ instance.

\begin{corollary}\label{cor:subsetsum} Suppose we are given $k$ arrays (for $k$
even) $A_1,\ldots, A_k$, each containing a collection of boolean vectors of $u$
bits.  There exists a polynomial-time algorithm which produces a formula $\psi$
of pathwidth $u+O(1)$ and a path decomposition of $\psi$ such that some
``special'' bag of the decomposition $B^*$ contains $u$ variables
$x_1,\ldots,x_u$ and such that we have the following: for each truth assignment
$\sigma$ to the variables of $B^*$, $\sigma$ can be extended to a satisfying
assignment of $\psi$ if and only if there exist $i_1,i_2,\ldots,i_k$ such that
$\sum_{j\in[k/2]} A_j[i_j] = (\sigma(x_1),\sigma(x_2),\ldots,\sigma(x_u)) =
\sum_{j\in[k/2+1,k]} A_j[i_j]$.  Therefore, $\psi$ is satisfiable if and only
if the given \kxor\ instance has a solution.\end{corollary}

\begin{proof}

We invoke Lemma~\ref{lem:subsetsum} on the \kxor\ instance $A_1,\ldots,A_{k/2}$ and
the \kxor\ instance $A_{k/2+1},\ldots,A_k$. This produces two formulas
$\phi_1,\phi_2$ and two path decompositions, where the last bags of the
decompositions have the properties described in the lemma. We obtain a formula
for the whole instance by taking the union of $\phi_1,\phi_2$ and identifying
the $u$ variables of the last bags (respecting the ordering). We construct a
path decomposition by merging the last bags of the two decompositions, and call
the bag resulting from this merging $B^*$.  \end{proof}

\subsubsection{Linear hash functions}\label{sec:hashalon}

We now recall some known properties of linear hash functions. In this section
(and following sections), when we say that we construct a random linear
function $h:\{0,1\}^u\to\{0,1\}^r$ we mean that we construct an $r\times u$
boolean matrix $R$ whose elements are set independently to 0 or 1 with
probability $1/2$. The function is then defined as $h(v)=Rv$. Furthermore, when
we write $h^{-1}(y)$, we mean the set of pre-images of $y$, that is,
$h^{-1}(y)=\{\ x\ |\ h(x)=y\ \}$. The following theorem informally states that
if we take a random linear function and apply it to a set of size equal to the
codomain of the function, every target value will have relatively few
pre-images, that is, there will be relatively few collisions per bucket.

\begin{theorem}[Theorem 4 of \cite{AlonDMPT99}]\label{thm:alon} Let $u,r$ be
integers, with $r<u$ and $\mathcal{H}$ be the set of all linear functions from
$\{0,1\}^u$ to $\{0,1\}^r$. For all sets $S\subseteq \{0,1\}^u$ such that
$|S|\le 2^r$ we have the following: $E_{h\in\mathcal{H}}[\max_{y\in\{0,1\}^r} |
h^{-1}(y)\cap S |] = O(r\log r)$.\end{theorem}

We remark that it was shown in \cite{AlonDMPT99} that in fact the same upper
bound on the maximum load can be shown even if $|S|$ is slightly larger than
$2^r$ (namely, the theorem also holds under the assumption $|S|<2^rr$).
Subsequently, Babka \cite{abs-1810-04161} gave an improvement to
Theorem~\ref{thm:alon} which upper bounds the expected maximum load by $O(r)$
(rather than $O(r\log r)$) and this was much more recently further improved by
Jaber, Kumar, and Zuckerman to match the performance of a truly random
function, that is, $O(r/\log r)$ \cite{JaberKZ25}.  Nevertheless, neither of
these improvements is necessary for our purposes and indeed
Theorem~\ref{thm:alon} is already more than good enough (it would have been
sufficient for us to have an upper bound of $2^{o(r)}$ on the maximum load).
We therefore stick with the bound of Theorem~\ref{thm:alon}. We also mention
that the requirement that $r<u$ covers all cases we are interested in, but for
the sake of completeness, the theorem can also be applied when $r\ge u$ by
simply padding the input vectors to increase their dimension.

A class of functions $\mathcal{H}$ from $\{0,1\}^u$ to $\{0,1\}^r$ is called
$1$-universal if  for all distinct $x_1,x_2\in\{0,1\}^u$ we have that
$Pr_{h\in\mathcal{H}}[h(x_1)=h(x_2)] \le \frac{1}{2^r}$. The following is a
standard fact about $1$-universal functions.

\begin{theorem}\label{thm:wegman} Let  $u,r$ be integers, with $r<u$ and
$\mathcal{H}$ be a $1$-universal set of functions from $\{0,1\}^u$ to
$\{0,1\}^r$. Then, for all sets $S\subseteq \{0,1\}^u$ such that $|S|\le
2^{r/2}$ we have the following: $Pr_{h\in\mathcal{H}}[\max_{y\in\{0,1\}^r} |
h^{-1}(y)\cap S| >1 ] \le \frac{1}{2}$. \end{theorem}

\begin{proof} There are at most ${|S|\choose 2}\le \frac{2^r}{2}$ possible
colliding pairs and each pair has probability at most $\frac{1}{2^r}$ of
actually colliding. Therefore, the expected number of collisions is at most
$\frac{1}{2}$. By Markov's inequality, the probability that there exists a pair
of distinct values that collide is at most $\frac{1}{2}$. \end{proof}

It is also known (and not too hard to observe) that the linear hash functions
we are using are indeed $1$-universal.

\begin{observation} Let $u,r$ be integers, with $r<u$ and $\mathcal{H}$ be the
set of all linear functions from $\{0,1\}^u$ to $\{0,1\}^r$. Then, for all
distinct $x_1,x_2\in\{0,1\}^u$ we have $Pr_{h\in\mathcal{H}}[h(x_1)=h(x_2)] \le
\frac{1}{2^r}$.\end{observation}

\begin{proof} It is sufficient to prove the observation for $r=1$, as in our
context a hash function with $r$ bits of output is just a concatenation of $r$
independent hash functions with one bit of output. Therefore if the probability
of collision is at most $1/2$ in this case, it is at most $\frac{1}{2^r}$ in
general. 

Let $R$ then be a random boolean vector of dimension $u$ and we want to show
that $Pr[Rx_1=Rx_2]\le \frac{1}{2}$ when $x_1\neq x_2$. Let $j$ be the last bit
in which $x_1,x_2$ differ. Let $x_2'$ be the vector obtained from $x_2$ by
flipping the $j$-th bit. Then $Pr[Rx_1=Rx_2] = \frac{1}{2}Pr[Rx_1=Rx_2'] +
\frac{1}{2}Pr[Rx_1\neq Rx_2']$. This follows because when $Rx_1=Rx_2'$, then
$Rx_1=Rx_2$ if and only if the $j$-th bit of $R$ is $0$ (which happens with
probability $1/2$); and similarly, when $Rx_1\neq Rx_2'$, then $Rx_1=Rx_2$ if
and only if the $j$-th bit of $R$ is $1$. However, $Pr[Rx_1=Rx_2']+Pr[Rx_1\neq
Rx_2']=1$.  \end{proof}

\subsection{The reduction}\label{sec:kxorreduction}

In this section we present a construction which reduces \kxor\ to \textsc{SAT}
in a way that will allow us to establish Theorem~\ref{thm:kxor}.  We are given a
\kxor\ instance: $k$ arrays (we assume without loss of generality that $k$ is
even) $A_1,\ldots,A_k$, each containing $n$ boolean vectors of $u$ bits each.
The question is whether there exist $i_1,i_2,\ldots,i_k$ such that
$\sum_{j\in[k]} A_j[i_j] = 0$. Equivalently, the question is whether there
exists a boolean vector $v$ with $u$ bits and $i_1,i_2,\ldots,i_k$ such that
$v=\sum_{j\in[k/2]} A_j[i_j] =\sum_{j\in[k/2+1,k]} A_j[i_j]$. We assume without
loss of generality that $n$ is a power of $2$.

Let $r_0=\frac{k}{2}\log n+1$ and pick a hash function $h^*$ uniformly at
random from the space of all linear functions from $\{0,1\}^u$ to
$\{0,1\}^{r_0}$. In the remainder we call $h^*$ our \emph{main} hash function.

Let $r_1=\lceil3\log(r_0)\rceil$. For each $\ell\in[10k\log n]$ pick a hash
function $h_\ell$ uniformly at random from the space of all linear functions
from $\{0,1\}^u$ to $\{0,1\}^{r_1}$. We will call these our \emph{secondary}
hash functions.

Our construction now proceeds as follows: for each $\ell\in[10k\log n]$ we
construct a new instance of \kxor\ made up of $k$ arrays of size $n$, where
each vector has $r_0+r_1$ bits. Let $B_1^\ell, B_2^\ell,\ldots, B_k^\ell$ be
the arrays of the new instance. We set for all $i\in[k], j\in[n]$ that
$B_i^\ell[j] = (h^*(A_i[j]), h_\ell(A_i[j]))$. In other words, to obtain
instance $\ell$ we apply to each element of the original instance the main hash
function as well as the $\ell$-th secondary hash function and replace the
element with the concatenation of the two returned values.

For each of the $10k\log n$ instances of \kxor\ we have produced we invoke
Corollary~\ref{cor:subsetsum}. This produces $10k\log n$ formulas
$\phi_1,\ldots,\phi_{10k\log n}$ and their corresponding path decompositions of
width $r_0+r_1+O(1)$. Recall that by Corollary~\ref{cor:subsetsum} each decomposition
contains a special bag $B^*$ with $r_0+r_1$ variables, such that an assignment
to this bag is extendible to a satisfying assignment to the formula if and only
if the vector of the assignment is constructible as a sum from the first $k/2$
arrays and also as a sum from the last $k/2$ arrays.  

So far we have constructed $10k\log n$ independent CNF formulas (and their path
decompositions). The last step of our construction adds some extra clauses to
ensure that all satisfying assignments must consistently select the same value
for the main hash function. More precisely, for each $\ell\in[10k\log n-1]$
consider the formulas $\phi_\ell$ and $\phi_{\ell+1}$ and the corresponding
special bags of their decompositions, which contain $r_0+r_1$ variables each.
We add clauses ensuring that the first $r_0$ variables of the special bag of
$\phi_\ell$ (encoding the value of the main hash function) must receive the
same assignments as the $r_0$ corresponding variables of the special bag of
$\phi_{\ell+1}$.

\subsection{Proof of correctness}\label{sec:kxorproof}

\begin{lemma}\label{lem:kxortreewidth} The algorithm of
Section~\ref{sec:kxorreduction} produces a formula of treewidth
$(1+o(1))\frac{k}{2}\log n$. Furthermore, a tree decomposition of this width
can be constructed in polynomial time.  \end{lemma}

\begin{proof} We start with the union of the path decompositions of the
formulas $\phi_1,\ldots,\phi_{10k\log n}$ guaranteed by Corollary~\ref{cor:subsetsum}.
These cover the whole construction except the clauses we added in the last step
to ensure consistency for the main hash function between different CNF
formulas. Let us then explain how to obtain a tree decomposition that also
covers these new clauses, while at the same time connecting the forest of path
decompositions into a single tree decomposition. 

For each $\ell\in[10k\log n-1]$ we do the following: consider the special bags
$B^*_\ell, B^*_{\ell+1}$ of the path decompositions of $\phi_\ell,
\phi_{\ell+1}$. The primal graph induced by these two bags is a matching with
$r_0$ edges, say from variables $x_1^\ell, x_2^\ell,\ldots, x_{r_0}^\ell$ to
variables $x_1^{\ell+1}, x_2^{\ell+1},\ldots,x_{r_0}^{\ell+1}$, because we have
added clauses (of arity two) ensuring that $x_j^\ell=x_j^{\ell+1}$. We connect
$B_\ell^*$ and $B_{\ell+1}^*$ in the decomposition by a sequence of $2r_0$ bags
as follows: start with $B_\ell^*$ as the current bag and then at each step
$j\in[r_0]$, add a bag that is a copy of the current bag with $x_j^{\ell+1}$
added, and then add a copy of the bag with $x_j^{\ell}$ removed and make that
the current bag. Continue in this way until the last bag contains all of
$x_j^{\ell+1}$, for $j\in[r_0]$, and make that bag a neighbor of
$B_{\ell+1}^*$. It is not hard to see that this constructs a valid tree
decomposition of the whole instance.

The width of this decomposition is dominated by the width of the path
decompositions given by Corollary~\ref{cor:subsetsum}, which is $r_0+r_1+O(1)$. Since
$r_0= \frac{k}{2}\log n +O(1)$ and $r_1=O(\log r_0) = o(k\log n)$ we have that
the width is at most $(1+o(1))\frac{k}{2}\log n$.  \end{proof}

We remark that the tree decomposition we obtain from the previous proof is in
fact not far from being a path decomposition, in the sense that the tree of the
decomposition itself has pathwidth $2$ (it is a sub-divided caterpillar).
Indeed, the construction could have been somewhat simplified by observing that
we have added clauses ensuring that the variables of the special bags remain
consistent; however, this could also be achieved by simply \emph{merging} these
variables. This would have resulted in a tree decomposition where the tree is a
union of paths with an extra bag adjacent to the middle of each path, which,
however, still has pathwidth $2$, so we still fall short of proving that the
\kxor\ hypothesis implies the \ppseth.  Therefore, we prefer to use the
decomposition of the previous proof, which makes applying
Corollary~\ref{cor:subsetsum} slightly cleaner.

\begin{lemma}\label{lem:kxorcorrect1} If the \kxor\ instance on which the
algorithm of Section~\ref{sec:kxorreduction} is executed has a solution, then the
resulting CNF formula is satisfiable. \end{lemma}

\begin{proof} If the given \kxor\ instance has a solution given by indices
$i_1,i_2,\ldots,i_k$, then there exists a boolean vector $v$ such that
$v=\sum_{j\in[k/2]}A_j[i_j]=\sum_{j\in[k/2+1,k]}A_j[i_j]$. We obtain a
satisfying assignment of the formula as follows.

For each $\ell\in[10k\log n]$ we assign to the $r_0+r_1$ variables of the
special bag $B^*_{\ell}$ the values $(h^*(v),h_{\ell}(v))$. Observe that since
we are using the same value for the first $r_0$ variables of the special bags,
this satisfies the consistency clauses added in the final step of the
construction.  Furthermore, by Corollary~\ref{cor:subsetsum} for each
$\ell\in[10k\log n]$ the assignment we have set is extendible to a satisfying
assignment of $\phi_\ell$.  This is because $\sum_{j\in[k/2]} B^\ell_j[i_j]
=\sum_{j\in[k/2+1,k]} B^\ell_j[i_j] = (h^*(v),h_\ell(v))$, as follows from the
linearity of the functions $h^*, h_\ell$.  \end{proof}

\begin{lemma}\label{lem:kxorcorrect2} If the \kxor\ instance  on which the
algorithm of Section~\ref{sec:kxorreduction} is executed has no solution, then the
resulting CNF formula is satisfiable with probability at most $o(1)$.
\end{lemma}

\begin{proof}

Suppose that the given \kxor\ instance has no solution. Let $S_1$ be the set of
vectors constructible from $A_1,\ldots,A_{k/2}$, that is, $v\in S_1$ if and only
if there exist $i_1,\ldots,i_{k/2}$ such that $v=\sum_{j\in[k/2]}A_j[i_j]$.
Similarly, let $S_2$ be the set of vectors constructible from
$A_{k/2+1},\ldots,A_k$ and we have by assumption that $S_1\cap S_2=\emptyset$.
Furthermore, $|S_1|,|S_2|\le n^{k/2}$.

Suppose now that we apply the main hash function $h^*$ to $S_1\cup S_2$. Since
$|S_1\cup S_2|\le 2n^{k/2}=2^{r_0}$, by Theorem~\ref{thm:alon} we have that
$E[\max_{y\in\{0,1\}^{r_0}} | (h^*)^{-1}(y)\cap (S_1\cup S_2) |] = O(r_0\log
r_0)$.  By using Markov's inequality we conclude that the probability that
there exists $y\in\{0,1\}^{r_0}$ such that $| (h^*)^{-1}(y)\cap (S_1\cup S_2) |
> r_0\log^2r_0$ is $o(1)$, that is, tends to $0$ as $n$ tends to infinity. In
the remainder we will therefore condition the analysis under the assumption
that for all $y\in\{0,1\}^{r_0}$ we have $| (h^*)^{-1}(y)\cap (S_1\cup S_2) | <
r_0\log^2r_0$.

By the constraints added in the last step of our construction of the CNF
formula we have that if a satisfying assignment exists, it must pick the same
values for the first $r_0$ variables of the special bags $B^*$ of each of the
$10k\log n$ instances.  Fix now a value $y\in\{0,1\}^{r_0}$. We want to show
that the probability that a satisfying assignment exists which uses $y$ as the
truth assignment for the first $r_0$ variables of the special bags is
$o(\frac{1}{n^{k/2}})$. Since there are $2^{r_0} = O(n^{k/2})$ such assignments
$y$, if we show this, then by union bound we will have that the probability
that a satisfying assignment exists is $o(1)$.

Consider now the set $S_y=(h^*)^{-1}(y)\cap (S_1\cup S_2)$. As mentioned, we
are assuming that $|S_y|\le r_0\log^2r_0$. Consider now a value
$\ell\in[10k\log n]$ and suppose that the secondary hash function $h_\ell$ is
collision-free for $S_y$, that is, for all $z\in\{0,1\}^{r_1}$ we have
$|h_\ell^{-1}(z)\cap S_y|\le 1$. We claim that in this case the CNF formula
indeed has no satisfying assignment which picks $y$ as the assignment to the
$r_0$ special variables.  Indeed, otherwise we would have a satisfying
assignment to the sub-formula $\phi_\ell$, which would by
Corollary~\ref{cor:subsetsum} correspond to a collection $i_1,\ldots,i_k$ such
that $\sum_{j\in[k/2]}B_j^\ell[i_j] = (y,z) =
\sum_{j\in[k/2+1,k]}B_j^\ell[i_j]$. In other words, we would have to find two
distinct elements of $S_y$ which have a collision for the secondary hash
function $h_\ell$, which is impossible as we assumed that this function is
collision-free on $S_y$.

We now observe that for each $\ell\in[10k\log n]$, the probability that
$h_\ell$ is collision-free on $S_y$ is at least $1/2$. Indeed, by
Theorem~\ref{thm:wegman}, if $|S_y|\le r_0\log^2r_0$, since the codomain of the
secondary hash function has size $2^{r_1} \ge r_0^3 > |S_y|^2$ for sufficiently
large $n$, we have that $h_\ell$ is not collision-free on $S_y$ with
probability at most $1/2$. However, for each $\ell\in[10k\log n]$, we pick each
$h_\ell$ independently and uniformly at random, so the probability that all the
secondary hash functions have collisions on $S_y$ is at most $2^{-10k\log n} =
n^{-10k}$. In particular, for a fixed assignment $y$ to the $r_0$ special
variables the constructed formula is satisfiable only if all secondary hash
functions have collisions, and this probability is at most
$n^{-10k}=o(n^{-k/2})$, therefore, by union bound the probability of any
satisfying assignment existing is $o(1)$.  \end{proof}

\kxorthm*

\begin{proof}

Suppose there exists an algorithm solving \textsc{SAT} parameterized by
treewidth in time $2^{(1-\eps)\tw}|\phi|^{C}$ for some constant $C$. We are
given a \kxor\ instance on $k$ arrays of $n$ vectors each and we execute the
construction of Section~\ref{sec:kxorreduction}. By Lemma~\ref{lem:kxorcorrect1} and
Lemma~\ref{lem:kxorcorrect2} deciding if the constructed formula is satisfiable is
equivalent (up to a very small probability of one-sided error) to deciding if
the original instance has a solution.

Let us then examine the running time of this procedure. The treewidth of the
new formula is, according to Lemma~\ref{lem:kxortreewidth}, at most
$(1+\frac{\eps}{2})\frac{k}{2}\log n$ (for sufficiently large $n$).
Furthermore, since the construction runs in randomized polynomial time, we have
$|\phi|=(kn)^{O(1)}$. Therefore, the running time of the whole procedure is at
most $2^{(1-\eps)(1+\frac{\eps}{2})\frac{k}{2}\log n}n^{C'}$, for some constant
$C'$.  We have $(1-\eps)(1+\frac{\eps}{2})<1-\frac{\eps}{2}$. Suppose then that
$k>k_0=\lceil \frac{8C'}{\eps}\rceil$.  Then $\frac{\eps k}{4} \ge \frac{\eps
k}{8}+C'$. We therefore have that the algorithm runs in time at most
$n^{(1-\frac{\eps }{8})\frac{k}{2}}$, disproving the hypothesis that \kxor\ is
hard.  \end{proof}

\section{\ksum\ hardness implies \twseth}\label{sec:ksum}

Our goal in this section is to establish the following theorem:

\begin{restatable}[\ksum\ hypothesis $\Rightarrow$
\twseth]{theorem}{ksumthm}\label{thm:ksum} Suppose there exists an $\eps>0$ and
an algorithm which, given a CNF formula $\phi$ and a tree decomposition of its
primal graph of width $\tw$ and size polynomial in $|\phi|$ decides if $\phi$
is satisfiable in time $(2-\eps)^\tw|\phi|^{O(1)}$.  Then, there exist
$\delta>0, k_0>0$ and a randomized algorithm which for all $k>k_0$ solves
\ksum\ on instances with $k$ arrays each containing $n$ integers in time
$n^{(1-\delta)\frac{k}{2}}$ with one-sided error and success probability
$1-o(1)$.  In other words, if the \twseth\ is false, then the randomized \ksum\
hypothesis is false.  \end{restatable}

The high-level strategy we will follow is the same as for
Theorem~\ref{thm:kxor}, but there are some additional technical obstacles we
need to overcome. In particular, the first ingredient will once again be a
reduction from (a variant of) \textsc{Subset Sum} with a target value that can
be encoded with $t$ bits to \textsc{SAT} on instances with pathwidth $t$. It is
therefore crucial to make $t$ small (that is, at most roughly $\frac{k}{2}\log
n$), which in turn seems to imply that we need to edit the given \ksum\
instance so that the largest absolute value is also of the order of $n^{k/2}$.

In general it is known that an arbitrary \ksum\ instance can be reduced (via a
randomized algorithm) to an equivalent one where integers have absolute values
$O(n^k)$ (Lemma 2.2 of \cite{DalirrooyfardLS25}). This is achieved by applying
an almost-linear hash function to the input integers. The almost-linearity of
the function guarantees that if $k$ integers sum to zero the sum of their
hashed values will also be almost zero, while it can be shown that the
probability of false positives ($k$ integers whose sum is non-zero, but whose
hashed values sum to zero) is small. We will use this approach, based on an
almost-linear (in fact, almost-affine) hash function due to Dietzfelbinger
\cite{Dietzfelbinger96,Dietzfelbinger18}.

The problem we are faced with now is that we cannot afford to allow the
integers of the hashed instance to have a range of size $n^k$, but we rather
need the range to be in the order of $n^{k/2}$. However, setting the range of
the hash function to be so small will almost surely produce many false
positives. 

Recall the strategy we used to deal with this problem for
Theorem~\ref{thm:kxor}: we reformulated the problem as a \textsc{List
Disjointness} question (if $L,R$ are the sets of at most $n^{k/2}$ integers
constructible from the first and last $k/2$ arrays respectively, is $L\cap
R\neq \emptyset$?) which can be solved by examining hashed values; then, a
purported solution is encoded by picking a specific hash value $y$ such that
$x_1\in L, x_2\in R$ and $h(x_1)=h(x_2)=y$; a key fact now was that even though
false positives exist in the hashed instance, once we fix a specific $y$ the
number of false positive pairs for this particular $y$ can be guaranteed to be
small ($\log ^{O(1)}n$), using a result of \cite{AlonDMPT99}
(Theorem~\ref{thm:alon}).  We then used a second-level hash function to
eliminate these (few) collisions.

The obstacle is now that a result analogous to Theorem~\ref{thm:alon} is not
known for the almost-affine hash functions we wish to use. Indeed, in our
setting we essentially want to hash a set of at most $N=n^{k/2}$ integers (the
sums constructible from the first $k/2$ arrays) into a range of size roughly
$N=n^{k/2}$. Whereas for a linear hash function of the type of
Theorem~\ref{thm:alon} this will with high probability produce a maximum load
poly-logarithmic in $N$, for almost-linear hash functions dealing with
integers, such as the ones we use, the best known upper bound on the maximum
load is $O(N^{1/3})$ given by Knudsen \cite{Knudsen19} and it is a known open
problem whether this can be improved further (Westover's work nicely summarizes
what is known \cite{abs-2307-13016}).  Unfortunately, this bound on the maximum
load is definitely not good enough for our purposes, because our plan is to use
a second-level hash function whose range will be quadratic in the maximum load,
but we cannot afford to increase the number of output bits (and therefore the
treewidth) by another $\log (N^{2/3}) = \frac{2k}{3}\log n$ bits. 

Thinking a little more carefully we observe that what we need for the strategy
of Theorem~\ref{thm:kxor} to work is a bound on the maximum load that is
stronger than what is known for our hash function, but could be weaker than the
bound of Theorem~\ref{thm:alon}. In particular, it would be good enough to show
that the maximum load produced by hashing $N$ elements in a range of size $N$
is at most sub-polynomial (that is, $N^{o(1)}$). Our main effort in this
section is then devoted to constructing an almost-affine integer hash function
which achieves this maximum load.

Let us then outline how our construction of an almost-affine hash function
achieving small maximum load works. The basic idea is simple: rather than using
Dietzfelbinger's hash function to map $N=n^{k/2}$ integers to integers with
$\frac{k}{2}\log n$ bits, we select independently $\frac{\log n}{2}$ hash
functions each of which produces as output an integer on $k$ bits. In other
words, we hash integers to integer \emph{vectors} of dimension $\frac{\log
n}{2}$, where each coordinate has $k$ bits. To give some intuition why we do
this, we note that Dietzfelbinger's hash functions are known to map elements in
a pairwise independent way.  This fact is of little help if we map $N$ items
into $N$ buckets, as the maximum load of a pairwise independent hash function
could be $\Theta(\sqrt{N})$; however, if we map $N$ elements into a small
number (say $2^k$) buckets, the maximum load is $\frac{N}{2^k}+O(\sqrt{N})\le
(1+o(1))\frac{N}{2^k}$, that is, with high probability every bit we are using
is decreasing the load by a factor of almost $2$. We can then analyze the
performance of the concatenation of the hash functions by using the fact that
each function succeeds with high probability in the sense we just described,
allowing us to show that with high probability most hash functions will work as
expected.  The precise details of this approach are given in
Lemma~\ref{lem:alon2}.

Given this approach we still need to solve a few secondary problems to apply the
strategy we used for \kxor. Namely, we need to adapt the reduction from
\textsc{Subset Sum} to \textsc{SAT} so that it works for vector sums (this is
straightforward); and we need to refine our analysis to take into account that
the functions we use are not linear but almost-affine. We review the basic
tools we will need in Section~\ref{sec:sumtools}; we then describe our
reduction from \ksum\ to \textsc{SAT} in Section~\ref{sec:ksumreduction}; and
finally we prove the correctness of our reduction, including the claims about
the maximum load of our main hash function, in Section~\ref{sec:ksumproof}.

\subsection{Tools}\label{sec:sumtools}

In this section we first give another variation of the reduction from
\textsc{Subset Sum} to \textsc{SAT} (similar to Lemma~\ref{lem:subsetsum})
which this time is able to handle integer vectors. We then review some facts we
will need about the hash functions we will be using, including how Chebyshev's
inequality and pairwise independence can be used to provide concentration
bounds on the load of any bucket (Lemma~\ref{lem:chebyshev}).

\subsubsection{Subset Sum to pathwidth-SAT}

Similarly to Section~\ref{sec:xorsat} we present a variant of the reduction
from \textsc{Subset Sum} to \textsc{SAT} of \cite{Lampis25} which is adapted to
our construction. In particular, we will focus on a version of \textsc{Subset
Sum} where the elements are vectors with integer coordinates, that is, elements
of $[0,2^u-1]^d$, where each integer coordinate is made up of $u$ bits and we
have $d$ coordinates for each vector. The question is whether it is possible to
select one vector from each set, so that their component-wise sum (modulo
$2^u$) is equal to a given target $t^*\in[0,2^u-1]^d$. The lemma below states
that it is possible to encode a \textsc{Subset Sum} instance into a CNF formula
with a path decomposition such that the last bag contains $du$ variables with
the following property: if we read these variables as encoding an element $t^*\in
[0,2^u-1]^d$, then the satisfying assignments of the formula correspond exactly
to the sums $t^*$ which can be realized in the instance.

For technical reasons, we formulate this lemma in a more general way, allowing
for the ranges of the integers in each component of the given vectors to be
distinct. We will therefore assume that we are dealing with $d$-dimensional
vectors, where for each $j\in[d]$ the integers of the $j$-th coordinate are in
the range $[0,2^{u_j}-1]$.

\begin{lemma}\label{lem:subsetsum2} Let $d$ be a positive integer,
$u_1,u_2,\ldots,u_d$ be $d$ positive integers and suppose we are given $\ell$
arrays $A_1,\ldots, A_\ell$, each containing a collection of elements of
$[0,2^{u_1}-1]\times[0,2^{u_2}-1]\times\ldots\times[0,2^{u_d}-1]$.  There
exists a polynomial-time algorithm which produces a formula $\psi$ of pathwidth
$(\sum_{j\in[d]}u_j)+O(1)$ and a path decomposition of $\psi$ such that the
last bag of the decomposition $B_0$ contains $t=\sum_{j\in[d]}u_j$ variables
$x_1,\ldots,x_{t}$ and such that we have the following: for each truth
assignment $\sigma$ to the variables of $B_0$, $\sigma$ can be extended to a
satisfying assignment of $\psi$ if and only if there exist
$i_1,i_2,\ldots,i_\ell$ such that $\sum_{j\in[\ell]} A_j[i_j] =
\left((\sigma(x_1)\sigma(x_2)\ldots\sigma(x_{u_1})),(\sigma(x_{u_1+1})\ldots\sigma(x_{u_1+u_2})),\ldots,(\sigma(x_{t-u_d+1})\ldots\sigma(x_{t}))\right)$,
where additions are component-wise and modulo $2^{u_j}$ for the $j$-th
component, $j\in[d]$.  \end{lemma}

\begin{proof}

The proof is essentially the same as for Lemma~\ref{lem:subsetsum}, so we
sketch the details and only focus on the differences. As in
Lemma~\ref{lem:subsetsum}, we will describe the construction for $\ell=1$; the
general construction can be obtained inductively by gluing together a formula
for the first $\ell-1$ arrays with the formula constructed for $A_\ell$.
Suppose then for the remainder that we have a single array $A$ containing $n$
vectors.

Similarly to Lemma~\ref{lem:subsetsum} we construct the variables
$x_{i,(j_1,j_2)}$ for $i\in[0,n], j_1\in[d], j_2\in[u_{j_1}]$; $s_i$ for
$i\in[n]$; and $y_i$ for $i\in[n]$. The intuitive meanings of these variables
are the same, namely, for a given $i\in[0,n]$, the $t=\sum_{j_1\in[d]} u_{j_1}$
variables $x_{i,(j_1,j_2)}$ encode the value of the sum we have constructed
after considering the first $i$ elements of the array; $s_i$ is set to $1$ if
we select the $i$-th element of the array; and $y_i$ is $1$ if at least one
element has been selected among the first $i$ elements of the array.

The clauses involving the $y_i$ and $s_i$ variables are identical to
Lemma~\ref{lem:subsetsum}. Furthermore, again as in Lemma~\ref{lem:subsetsum}
we add clauses ensuring that if $s_i$ is $0$, then
$x_{i-1,(j_1,j_2)}=x_{i,(j_1,j_2)}$, for all $j_1\in[d], j_2\in[u_{j_1}]$.

The main difference with the proof of Lemma~\ref{lem:subsetsum} is the clauses
we add to cover the case when $s_i$ is set to $1$. In this case we have to
ensure that the value encoded by $x_{i,(j_1,j_2)}$ is equal to the value
encoded by $x_{i-1,(j_1,j_2)}$ plus $A[i]$ and this is slightly more
complicated than in Lemma~\ref{lem:subsetsum} because rather than bitwise XOR
we have to compute integer additions, which means we have to take into account
the calculation of carry bits.

For each $i\in[n]$ and for each $j_1\in[d]$ we construct $u_{j_1}$ ``carry''
variables $c_{i,(j_1,j_2)}$, $j_2\in[u_{j_1}]$. We denote by $A[i]_{j_1,j_2}$
the $j_2$-th least significant bit of the $j_1$-th coordinate of $A[i]$. We add
clauses ensuring that if $s_i$ is set to $1$, then for all $j_1\in[d]$ we have
(i) $x_{i,(j_1,1)} = x_{i-1,(j_1,1)}\oplus A[i]_{j_1,1}$ (ii) $c_{i,(j_1,1)} =
x_{i-1,(j_1,1)}\land A[i]_{j_1,1}$ (iii) for all $j_2\in[2,u_{j_1}]$ we ensure
that $x_{i,(j_1,j_2)} = x_{i-1,(j_1,j_2)}\oplus A[i]_{j_1,j_2}\oplus
c_{i,(j_1,j_2-1)}$ and $c_{i,(j_1,j_2)} = 1$ if and only if at least two of
$x_{i-1,(j_1,j_2)}, A[i]_{j_1,j_2}, c_{i,(j_1,j_2-1)}$ are set to $1$. 

We now observe that the clauses added above implement the desired property,
namely, that when $s_i$ is set to $1$ we must give values to
$x_{i-1,(j_1,j_2)}$ and $x_{i,(j_1,j_2)}$ such that the difference of their
encoded values is $A[i]$. Furthermore, it is not hard to see that we can build
a path decomposition of the new instance in a manner similar to that of
Lemma~\ref{lem:subsetsum} with the only difference that in the bags that
contain $x_{i-1,(j_1,j_2)}$ and $x_{i,(j_1,j_2)}$ we also add the carry
variables $c_{i,(j_1,j_2-1)}$ and $c_{i,(j_1,j_2)}$.  \end{proof}

We now obtain a corollary similar to Corollary~\ref{cor:subsetsum} by gluing
together the formulas representing the first half and second half of the arrays
of a \ksum\ instance.

\begin{corollary}\label{cor:subsetsum2} Let $d$ be a positive integer,
$u_1,\ldots,u_d$ be $d$ positive integers and suppose we are given $k$ arrays
(for $k$ even) $A_1,\ldots, A_k$, each containing a collection of integer
vectors from $[0,2^{u_1}-1]\times
[0,2^{u_2}-1]\times\ldots\times[0,2^{u_d}-1]$.  There exists a polynomial-time
algorithm which produces a formula $\psi$ of pathwidth
$(\sum_{j\in[d]}u_j)+O(1)$ and a path decomposition of $\psi$ such that some
``special'' bag of the decomposition $B^*$ contains $t=\sum_{j\in[d]}u_j$
variables $x_1,\ldots,x_{t}$ and such that we have the following: for each
truth assignment $\sigma$ to the variables of $B^*$, $\sigma$ can be extended
to a satisfying assignment of $\psi$ if and only if there exist
$i_1,i_2,\ldots,i_k$ such that $\sum_{j\in[k/2]} A_j[i_j] =
\left((\sigma(x_1)\sigma(x_2)\ldots\sigma(x_{u_1})),(\sigma(x_{u_1+1})\ldots\sigma(x_{u_1+u_2})),\ldots,(\sigma(x_{t-u_d+1})\ldots\sigma(x_{t}))\right)
= \sum_{j\in[k/2+1,k]} A_j[i_j]$.  \end{corollary} 

We remark that Corollary~\ref{cor:subsetsum2} can be interpreted as follows:
the produced formula $\psi$ is satisfiable if and only if the \ksum\ instance
obtained from $A_1,\ldots,A_k$ by multiplying all integers in the arrays
$A_{k/2+1},\ldots,A_k$ by $-1$ has a solution. As a result, the way we will use
the corollary is to take an arbitrary \ksum\ instance, multiply the second half
of the arrays by $-1$, and then invoke the corollary to produce a formula
$\psi$ equivalent to the original instance.

\begin{proof} Identically to the proof of Corollary~\ref{cor:subsetsum}, we
produce formulas $\phi_1,\phi_2$ from $A_1,\ldots,A_{k/2}$ and
$A_{k/2+1},\ldots,A_k$ respectively, and then identify the variables of the
last bags, obtaining the special bag $B^*$.  \end{proof}

\subsubsection{Linear and affine hash functions}\label{sec:hashsum}

The main hash function we will rely on is the following function due to
Dietzfelbinger \cite{Dietzfelbinger96,Dietzfelbinger18}. Let
$U=\{0,\ldots,u-1\}$ be a universe, $M=\{0,\ldots,m-1\}$ a target range, and
$r=um$, where $u,m$ are positive integers with $m<u$. Then we define the hash
function $h_{a,b}(x) = \left\lfloor \frac{(ax+b) \bmod r}{u} \right\rfloor$ and
the class of hash functions $\mathcal{H}_{u,m} = \{ h_{a,b}\ |\ 0\le a,b <r
\}$. The following theorem was shown (in a more general form) in
\cite{Dietzfelbinger96,Dietzfelbinger18}:

\begin{theorem}\label{thm:dietz} If $u,m$ are powers of $2$, then
$\mathcal{H}_{u,m}$ is a pairwise independent family of hash functions, that
is, for all distinct $x_1,x_2\in U$ and for all $y_1,y_2\in M$ we have
$Pr_{h\in\mathcal{H}_{u,m}}[ h(x_1)=y_1\land h(x_2)=y_2] = \frac{1}{|M|^2}$.
\end{theorem}

In the remainder, whenever we use this family of hash functions we will always
have that $u,m$ are powers of $2$.  Observe that Theorem~\ref{thm:dietz}
immediately implies that the family $\mathcal{H}_{u,m}$ is $1$-universal, so
Theorem~\ref{thm:wegman} applies. However, using Chebyshev's inequality, we can
also establish the following:

\begin{lemma}\label{lem:chebyshev}  Let $u,m$ be powers of $2$. Then, for all
$\delta>0$, $S\subseteq U$ and $y\in M$ we have $Pr_{h\in\mathcal{H}_{u,m}}[
|h^{-1}(y)\cap S| \ge (1+\delta)\frac{|S|}{|M|} ] \le \frac{|M|}{\delta^2|S|}$
\end{lemma}

\begin{proof} Let $Y=|h^{-1}(y)\cap S|$ be the random variable we are
interested in. We can see $Y$ as the sum of $|S|$ Bernoulli random variables
where for each $x\in S$ the corresponding variable takes value $1$ if $h(x)=y$,
that is, with probability $\frac{1}{|M|}$. We have $E[Y] = \frac{|S|}{|M|}$ and
$Var[Y] = \frac{|S|}{|M|}(1-\frac{1}{|M|}) \le E[Y]$ where we used the fact
that the variance of the sum of pairwise independent variables is equal to the
sum of their variances.  By Chebyshev's inequality we have that $Pr[ Y\ge E[Y]
+ \delta E[Y] ] \le \frac{Var[Y]}{\delta^2 (E[Y])^2} \le \frac{|M|}{\delta^2
|S|}$.  \end{proof}

Furthermore, $1$-universality will not be sufficient for our \ksum\
construction, because the functions we are using are not exactly linear (or
affine).  We will therefore need the following variation of
Theorem~\ref{thm:wegman} which states informally that if the range we are
mapping is $k^2$ times larger than what is needed for Theorem~\ref{thm:wegman},
then not only do we have no collisions but in fact distinct keys are mapped
into buckets at distance at least $k$ from each other.

\begin{lemma}\label{lem:wegman2} Let $u,m$ be powers of $2$ and $k$ a positive
integer. Then, for all sets $S\subseteq [0,u-1]$ with $|S| \le
\frac{\sqrt{m}}{k}$ we have the following:
$Pr_{h\in\mathcal{H}_{u,m}}[\max_{y\in[0,m-1]} |S\cap (\bigcup_{j\in[0,k-1]}
h^{-1}( (y+j) \bmod m))|>1]  \le \frac{1}{2}$.  \end{lemma}

\begin{proof}

Fix a $y\in[0,m-1]$ and consider the values $(y+j)\bmod m$, for $j\in[0,k-1]$.
Consider two distinct elements $x_1,x_2\in S$. The probability that
$h(x_1)=y+j_1\bmod m$ and $h(x_2)=y+j_2\bmod m$, for some given
$j_1,j_2\in[0,k-1]$ is at most $\frac{1}{m^2}$, by Theorem~\ref{thm:dietz}. By
union bound, the probability that there exist $j_1,j_2\in[0,k-1]$ such that
$h(x_1)=y+j_1\bmod m$ and $h(x_2)=y+j_2\bmod m$ is at most $\frac{k^2}{m^2}$.
Again by union bound, the probability that there exist distinct $x_1,x_2$ such
that there exist $y_1,y_2$ as stated is at most ${|S|\choose 2}\frac{k^2}{m^2}
\le \frac{|S|^2}{2} \frac{k^2}{m^2} \le \frac{1}{2m}$. Taking a union bound
over the $m$ possible values of $y$ completes the proof.  \end{proof}

It will at some point be more convenient to work with the functions
$h_{a,b}(x)$ of $\mathcal{H}_{u,m}$ that have offset $b=0$. Given a function
$h(x)=\lfloor \frac{(ax+b)\bmod r}{u}\rfloor$ we define $\hat{h}(x) = \lfloor
\frac{ (ax)\bmod r}{u} \rfloor$, that is, $\hat{h}$ is the modification of
$h(x)$ where we have removed the offset $b$. The functions $\hat{h}$ are
almost-linear in the following sense:

\begin{observation}\label{obs:almost} For all $h\in\mathcal{H}_{u,m}$ and
$x_1,x_2\in U$ we have that $ \hat{h}(x_1+x_2)\equiv \hat{h}(x_1) +
\hat{h}(x_2)+ \kappa' \mod m $, where $\kappa'\in\{0,1\}$.  \end{observation}

\begin{proof}

Let $\hat{h}(x)=\lfloor \frac{(ax)\bmod r}{u}\rfloor$.  Then, $\hat{h}(x_1+x_2)
= \lfloor \frac{(ax_1+ax_2)\bmod r}{u} \rfloor$, while
$\hat{h}(x_1)+\hat{h}(x_2)= \lfloor \frac{(ax_1)\bmod r}{u} \rfloor + \lfloor
\frac{(ax_2)\bmod r}{u} \rfloor = \lfloor \frac{(ax_1)\bmod r + (ax_2)\bmod
r}{u} \rfloor - \kappa$, where $\kappa\in \{0,1\}$. We now observe that the
difference between $(ax_1)\bmod r+(ax_2)\bmod r$ and $(ax_1+ax_2)\bmod r$ is an
integer multiple of $r$; therefore, dividing this difference by $u$ gives an
integer multiple of $m$ (since $r=mu$); therefore the two values are equal if
calculations are performed modulo $m$.  \end{proof}

\subsection{The reduction}\label{sec:ksumreduction}

We are given a \ksum\ instance: $k$ arrays (we assume without loss of
generality that $k$ is even) $A_1,\ldots,A_k$, each containing $n$ integers.
The question we want to answer is whether there exist $i_1,\ldots,i_k$ such
that $\sum_{j\in[k]} A_j[i_j]=0$. Equivalently, if we multiply all entries of
$A_{k/2+1},\ldots,A_k$ by $-1$ the question becomes whether there exists an
integer $v$ and $i_1,\ldots, i_k$ such that $v=\sum_{j\in[k/2]} A_j[i_j] =
\sum_{j\in[k/2+1,k]} A_j[i_j]$. We will focus on the latter formulation of the
problem and assume without loss of generality that $n$ is a power of $4$.
Furthermore, by adding $\max_{j\in[k], i\in[n]}|A_j[i]|$ to all entries we can
now assume that all integers are non-negative without affecting the answer. Let
$W=\max_{j\in[k], i\in[n]}A_j[i]$ be the largest integer of the resulting
instance.

Similarly to the proof of Theorem~\ref{thm:kxor} we will now construct a
\emph{main} hash function.  However, our main hash function will now be a
concatenation of many hash functions with smaller range. Let $m=2^k$ and let
$u$ be the smallest power of $2$ such that $kW <u$.  For $\beta\in[\frac{\log
n}{2}]$ pick independently and uniformly at random a hash function
$h^*_\beta\in\mathcal{H}_{u,m}$. The main hash function is the concatenation
$h^*=(h^*_1,h^*_2,\ldots,h^*_{\frac{\log n}{2}})$.  Observe that if we apply
$h^*$ to an integer in the input we obtain a vector of $\frac{\log n}{2}$
integers, each of $k$ bits, so $\frac{k}{2}\log n$ bits of output.

Let $\delta>0$ be a small fixed constant (we will precisely define $\delta$ in
the proof of correctness in a way that will depend on the running time of the
supposed algorithm falsifying the \twseth). Let $m'$ be the smallest power of
$2$ such that $m'>4k^2n^{4\delta k}$. For each $\ell\in[10k\log n]$ pick a
(secondary) hash function $h_\ell$ independently and uniformly at random from
$\mathcal{H}_{u,m'}$. We note in passing that the main and secondary hash
functions are well-defined, as $u>m$ and $u>m'$ if we set $\delta$ to be a
small fixed constant.  This follows because we can assume without loss of
generality that $W>n^{k/3}$, as otherwise a simple $O(nW)$ dynamic programming
alogrithm solves \ksum\ in time much faster than $n^{k/2}$.

For each $\ell\in[10k\log n]$ we now construct a distinct (hashed) instance of
\ksum, using the main hash function $h^*=(h^*_1,h^*_2,\ldots,h^*_{\frac{\log
n}{2}})$ and the secondary hash function $h_\ell$. In order to do so, we will
need to take into account that the functions we are using are in fact not
linear but almost-affine transformations, that is, each function $h_{a,b}(x)$
adds to the output a constant offset $b$ and then taking the floor of the
result may introduce a small error.

Recall that for a hash function $h_{a,b}(x) = \lfloor \frac{(ax+b)\bmod r}{u}
\rfloor$ we defined $\hat{h}_{a,b}(x) = h_{a,0}(x) = \lfloor \frac{ax\bmod
r}{u} \rfloor$, that is, the function we obtain if we remove the offset $b$.
For the main hash function $h^*$ we set
$\hat{h}^*=(\hat{h}^*_1,\ldots,\hat{h}^*_{\frac{\log n}{2}})$. We construct the
$\ell$-th \ksum\ instance by constructing $k$ arrays $B_1^\ell,\ldots,B_k^\ell$
and setting for all $j\in[k], i\in[n]$ that $B_j^\ell[i] =
(\hat{h}^*(A_j[i]),\hat{h}_\ell(A_j[i]))$. Furthermore, we add to the instance
two more arrays $B_0^\ell, B_{k+1}^\ell$ as follows. Let $b_1^*, b_2^*,\ldots,
b_{\frac{\log n}{2}}^*$ be the offset values of $h^*_1,h^*_2,\ldots,
h^*_{\frac{\log n}{2}}$ respectively, and $b_\ell$ be the offset value of
$h_\ell$. Then, we add to $B_0^\ell$ and to $B_{k+1}^\ell$ all vectors
$(\lfloor \frac{b_1^*}{u}\rfloor ,\lfloor \frac{b_2^*}{u}\rfloor,\ldots,
\lfloor\frac{b^*_{{\log n}/{2}}}{u}\rfloor,\lfloor\frac{b_\ell}{u}\rfloor)+v$,
where $v\in [0,k]^{\frac{\log n}{2}+1}$.

We now invoke Corollary~\ref{cor:subsetsum2} on each of the $[10k\log n]$
instances we constructed to produce $10k\log n$ CNF formulas and corresponding
decompositions of pathwidth $\frac{k}{2}\log n+4\delta k\log n+2\log k+O(1)$,
where the bound on the pathwidth follows because the output of the main hash
function consists of $\frac{\log n}{2}$ integers on $k$ bits and the output of
each secondary hash function is an integer on $\log m' \le \log (8k^2n^{4\delta
k})$ bits. Finally, in the same way as for \kxor, for each $\ell\in[10k\log
n-1]$ we add constraints between the first $\frac{k}{2}\log n$ variables of the
special bag $B^*_\ell$ and the special bag $B^*_{\ell+1}$ to ensure that a
consistent value needs to be selected for the main hash function across all
instances. This completes the construction.

\subsection{Proof of correctness}\label{sec:ksumproof}

We now present the proof of correctness of the construction of the previous
section, following the same outline as for Theorem~\ref{thm:kxor}. The main
difference is that before establishing the ``hard'' part of the reduction (if
the formula is satisfiable, then with high probability there is a \ksum\
solution, given in Lemma~\ref{lem:ksumred2}) we need to prove that the expected
maximum load of our main hash function is sub-polynomial, as promised. This is
done in Lemma~\ref{lem:alon2}.

\begin{lemma}\label{lem:ksumtw} The algorithm of
Section~\ref{sec:ksumreduction} produces a CNF formula with primal treewidth at
most $\frac{k}{2}\log n + 4\delta k\log n +O(\log k)$, as well as a
decomposition of this width.  \end{lemma}

\begin{proof} As mentioned, when we invoke the algorithm of
Corollary~\ref{cor:subsetsum2} we obtain path decompositions of width
$\frac{k}{2}\log n + 4\delta k\log n +O(\log k)$. The extra constraints we
added to ensure consistency between instances can be handled in a tree
decomposition in the same way as in the proof of Lemma~\ref{lem:kxortreewidth}.
\end{proof}

\begin{lemma}\label{lem:ksumred1} If the original \ksum\ instance on which the
algorithm of Section~\ref{sec:ksumreduction} is applied has a solution, then
the resulting formula is satisfiable. \end{lemma}

\begin{proof} 

This lemma is slightly more complicated to establish than
Lemma~\ref{lem:kxorcorrect1} because the hash functions we use are not exactly
linear, but the floor operation may introduce a small error. This is the reason
why we added to our construction the extra arrays $B_0^{\ell}, B_{k+1}^\ell$,
which allow us to add to the input a desired error correction value.

Let us be more precise. Suppose there is a solution $i_1,i_2,\ldots,i_k$ to the
given \ksum\ instance such that $\sum_{j\in[k/2]}A_j[i_j] =
\sum_{j\in[k/2+1,k]} A_j[i_j] = v$. We will attempt to construct a satisfying
assignment by giving a value encoding $h^*(v)$ to the first $\frac{k}{2}\log n$
variables of the special bags of the $10k\log n$ formulas. This satisfies the
consistency constraints added in the end, so we need to show that for all
$\ell\in[10k\log n]$ this can be extended to a satisfying assignment of the
formula representing the $\ell$-th \ksum\ instance.

We complete the assignment to the special bag $B^*_\ell$ by giving value
$h_\ell(v)$ to the remaining variables. We now claim there exist $i_0$ and
$i_{k+1}$ such that $\sum_{j\in[0,k/2]}B_j^\ell[i_j] =
\sum_{j\in[k/2,k+1]}B_j^\ell[i_j] = (h^*(v),h_\ell(v))$. In that case, there
exists a satisfying assignment by Corollary~\ref{cor:subsetsum2}.

A potential difficulty now is that because the functions $\hat{h}$ are not
exactly linear, the sum $\sum_{j\in[k/2]}B^{\ell}_j[i_j]$ may not exactly
correspond to the target value because (i) we have used the functions $\hat{h}$
to produce the entries of $B^{\ell}_j$, so the offset of our hash functions is
missing (ii) the functions $\hat{h}$ are not exactly linear.

However, by invoking Observation~\ref{obs:almost} we can see that the functions
$\hat{h}$ have the property that for any $k$ integers $x_1,x_2,\ldots,x_k$ we
have $\hat{h}(\sum_{i\in[k]} x_i) \equiv \sum_{i\in[k]} \hat{h}(x_i) + \kappa'
\mod m$, where $\kappa'\in [0,k-1]$. Because $B^\ell_0$ contains entries which
are equal to the missing offset vectors plus any vector from $[0,k]^{\frac{\log
n}{2}+1}$, it is always possible to select a vector $B^{\ell}_0[i_0]$ (and
similarly from $B^{\ell}_{k+1}$) so that we can construct the target sum and
hence satisfy the formula.  \end{proof} 

%
%
%
%

Before we proceed with Lemma~\ref{lem:ksumred2} which establishes the more
challenging direction of the reduction we need to establish that the main hash
function we are using has the desired property that with high probability no
bucket has very high load. This is shown in the following lemma which states
that if we select $\frac{\log n}{2}$ hash functions with $k$ bits of output
each (as we do in our construction) the probability that there exists a set of
size $n^{\delta k}$ which collide on all hash functions is very small. In other
words, the expected maximum load for our main hash function is $n^{o(1)}$ with
high probability.

\begin{lemma}\label{lem:alon2} For all $S\subseteq [0,u-1]$ with $|S|\le
2n^{k/2}$ and for all $\delta>0$ we have the following: suppose we select
independently and uniformly at random $\frac{\log n}{2}$ hash functions
$h_1,\ldots,h_{\frac{\log n}{2}}\in \mathcal{H}_{u,m}$ with $m=2^k$. Then, the
probability there exists $(y_1,y_2,\ldots,y_{\frac{\log n}{2}})\in
[0,2^k-1]^{\frac{\log n}{2}}$ such that $|\bigcap_{j\in\frac{\log n}{2}}
(h_j^{-1}(y_j) \cap S)| \ge n^{\delta k}$ tends to $0$ as $n$ tends to
infinity.  \end{lemma}

\begin{proof}

Fix a $\delta>0$ and a specific vector $(y_1,y_2,\ldots,y_{\frac{\log
n}{2}})\in [0,2^k-1]^{\frac{\log n}{2}}$. We will show that for any such
vector, the probability that $|\bigcap_{j\in\frac{\log n}{2}} (h_j^{-1}(y_j)
\cap S)| \ge n^{\delta k}$ is $o(\frac{1}{n^{k/2}})$. Since there are $n^{k/2}$
such vectors, by union bound we will have that the probability that a vector
satisfying the stated property exists tends to $0$.

Let $S_0=S$ and for $b\in [\frac{\log n}{2}]$ we define a random variable $S_b$
as $S_b = \bigcap_{j\in [b]} (h_j^{-1}(y_j) \cap S)$. In other words, $S_b$ is
the set of elements of $S$ which are hashed to a vector that is consistent with
the $b$ first components of our selected vector. We want to argue that the
probability that $|S_{\frac{\log n}{2}}|\ge n^{\delta k}$ is small.

Define for each $b\in [\frac{\log n}{2}]$ a random variable $X_b$ which is
equal to $1$ if $|S_b|\le (1+\frac{\delta}{4})\frac{|S_{b-1}|}{2^k}$. We have
that $Pr[X_b = 0\ |\ S_{b-1}] \le \frac{2^{k+4}}{\delta^2 |S_{b-1}|}$ by
Lemma~\ref{lem:chebyshev}, that is, we can upper bound the probability of
$X_b=0$ conditioned on any currently surviving set of collisions $S_{b-1}$.  We
now distinguish two cases:

\begin{enumerate}

\item Suppose $\sum_{b\in[\frac{\log n}{2}]} X_b > \frac{\log n}{2} -
\frac{10}{\delta}$.  Observe that for all $b\in[\frac{\log n}{2}]$ we have
$|S_b|\le |S_{b-1}|$ and furthermore, if $X_b=1$ then $|S_b|\le
(1+\frac{\delta}{4})\frac{|S_{b-1}|}{2^k}$. We therefore have $|S_{\frac{\log
n}{2}}| \le |S| \cdot \left( \frac{1+\frac{\delta}{4}}{2^k} \right)^{\frac{\log
n}{2}-\frac{10}{\delta}} \le \left( 1+\frac{\delta}{4} \right)^{\frac{\log
n}{2}}\cdot 2^{\frac{10k}{\delta}+1} \le e^{\frac{\delta \log n}{8}}
2^{\frac{10k}{\delta}+1} \le n^{\frac{\delta}{2}}$ where we used that $|S|\le
2n^{k/2}=2\cdot (2^k)^{\frac{\log n}{2}}$, that $1+\frac{\delta}{4}\le
e^{\frac{\delta}{4}}$, that $e^{\log n}<n^2$, and that
$2^{\frac{10k}{\delta}+1}<n^{\frac{\delta}{4}}$ for sufficiently large $n$. So
in this case $S_{\frac{\log n}{2}}$ is as small as desired.

\item The remaining case is $\sum_{b\in[\frac{\log n}{2}]} X_b \le \frac{\log
n}{2} - \frac{10}{\delta}$, therefore there exist at least $\frac{10}{\delta}$
distinct variables $X_b$ which have value $0$. We claim that the probability of
this happening while $S_{\frac{\log n}{2}}$ still has large size is very small.
Suppose then that $|S_{\frac{\log n}{2}}|\ge n^{\delta k}$, which implies that
$|S_b|\ge n^{\delta k}$ for all $b\in[\frac{\log n}{2}]$. But then, for all
$b\in[\frac{\log n}{2}]$ we have $Pr[X_b = 0\ |\ S_{b-1}] \le
\frac{2^{k+4}}{\delta^2 n^{\delta k}}$. Since this upper bound applies to the
conditional probability of $X_b=0$, the probability that there exist at least
$10/\delta$ indices $b\in[\frac{\log n}{2}]$ with $X_b=0$ is at most $\left(
\frac{2^{k+4}}{\delta^2 n^{\delta k}}\right)^{10/\delta}$ and this value is at
most $o(\frac{1}{n^k})$ as desired.

%

\end{enumerate}

In other words, the above analysis establishes that either almost all hash
functions are successful (in the sense that they place at most an
$(1+\frac{\delta}{4})$ fraction of the remaining items on the bucket we care
about), in which case in the end our bucket contains few items; or many hash
functions fail. However, by Lemma~\ref{lem:chebyshev} this can happen with
probability at most (roughly) $n^{-\delta k}$ if the current collection of
items is large (larger than $n^{\delta k}$), so it is very unlikely that more
than $\frac{10}{\delta}$ hash functions fail in this way.  \end{proof}

\begin{lemma}\label{lem:ksumred2} Suppose that the algorithm of
Section~\ref{sec:ksumreduction} is applied to an instance of \ksum\ that has no
solution. If $2\log k<\delta k$, then the probability that a satisfiable
formula is output is $o(1)$.  \end{lemma}

\begin{proof} 

We follow a strategy similar to the proof of Lemma~\ref{lem:kxorcorrect2} with
the main difference being that we need to account for the fact that the use of
the floor function in our hash introduces some minor errors, as seen in
Observation~\ref{obs:almost}.  Suppose that the given \ksum\ instance has no
solution and let $S_1,S_2$ respectively be the sums constructible from
$A_1,\ldots,A_{k/2}$ and $A_{k/2+1},\ldots,A_k$ respectively. By assumption
$S_1\cap S_2=\emptyset$.  Also, $|S_1|,|S_2|\le n^{k/2}$.

Our main hash function is made up of $\frac{\log n}{2}$ independently chosen
hash functions, as required by Lemma~\ref{lem:alon2}. We will therefore
condition our analysis on the event that the statement of Lemma~\ref{lem:alon2}
is satisfied, that is, for all $y=(y_1,y_2,\ldots,y_{\frac{\log n}{2}})$, there
are at most $n^{\delta k}$ elements of $S_1\cup S_2$ mapped to $y$ by $h^*$.

Suppose that we want to find a satisfying assignment to the produced CNF
formula and observe that because of the consistency constraints added in the
last step we must select a consistent value to the $\frac{k}{2}\log n$
variables that represent the output of the main hash function $h^*$. Fix then
such an assignment $y=(y_1,\ldots,y_{\frac{\log n}{2}})$ and we will show that
the probability that this assignment can be extended to an assignment that
satisfies the formula is $o(\frac{1}{n^{k/2}})$ and the lemma will follow by
union bound.

Consider all vectors $\tilde{y}$ which are ``close'' to $y$, meaning that each
coordinate differs from the corresponding coordinate of $y$ by at most $2k$
(here differences are calculated cyclically modulo $m$).  The set of all such
vectors is at most $(4k+1)^{\frac{\log n}{2}} < n^{2\log k}$.  Since we are
conditioning on the event that no vector has more than $n^{\delta k}$ elements
of $S_1\cup S_2$ mapped to it, the total number of elements of $S_1\cup S_2$
which are mapped to a vector $\tilde{y}$ that is close to $y$ is at most
$n^{\delta k + 2\log k} < n^{2\delta k}$.

We now observe that if we fix $y$ as the assignment encoding the output of
$h^*$, by Corollary~\ref{cor:subsetsum2} for all $\ell\in[10k\log n]$ in order
to produce a satisfying assignment to the formula representing the $\ell$-th
\ksum\ instance, we must select elements of $S_1,S_2$ which are mapped by $h^*$
close to $y$. This can be seen because by an analysis similar to that of
Lemma~\ref{lem:ksumred1} the sums of the hashed values differ from the hashed
value of a sum by at most $k/2$ (using Observation~\ref{obs:almost}) and the
extra elements of $B_0^{\ell}$ allow us to modify a sum by at most $k$.
Therefore, there are at most $n^{2\delta k}$ elements of $S_1\cup S_2$ which
correspond to possible satisfying assignments of the $\ell$ \ksum\ instances
(extending the assignment encoding $y$).

Now we recall that we have set up the range of the secondary hash function to
be at least $4k^2n^{4\delta k}$, meaning that, according to
Lemma~\ref{lem:wegman2} each secondary hash function is collision-free with
probability at least $\frac{1}{2}$ on the set of elements of $S_1\cup S_2$
which may lead to a satisfying assignment extending the current one. More
strongly, Lemma~\ref{lem:wegman2} establishes not only collision-freeness but
the property that any two distinct elements are mapped to values at distance
more than $2k$. If there exists a secondary hash function that achieves this
strong collision-freeness property, then the corresponding \ksum\ instance
cannot be satisfied while extending the current assignment (because using
Observation~\ref{obs:almost} we can bound the error due to non-linearity by at
most $k/2$).  We now observe that the probability that no secondary hash
function achieves the strong collision-freeness property is at most
$2^{-10k\log n}= o(\frac{1}{n^k})$.  Hence, the probability that the assignment
encoding $y$ can be extended to a satisfying assignment to the whole instance
is small, as desired.  \end{proof}

\ksumthm*

\begin{proof} Suppose that we have an algorithm solving \textsc{SAT} in time
$2^{(1-\eps)\tw}|\phi|^{O(1)}$. We set $\delta=\frac{\eps}{40}$ and $k_0$
sufficiently large (to be defined later). We now execute the algorithm of
Section~\ref{sec:ksumreduction} to obtain an instance which preserves the
answer with high probability according to Lemmas~\ref{lem:ksumred1} and
\ref{lem:ksumred2}. 

Suppose the original instance has $k$ arrays of $n$ integers. We constructed
$O(k\log n)$ hashed instances, with total size $O(k^{O(\log n)}) = n^{O(\log
k)}$ (this is dominated by the arrays $B^\ell_0, B^\ell_{k+1}$ which contain a
vector for each $[0,k]^{\log n/2}$). The algorithm of
Corollary~\ref{cor:subsetsum2} runs in polynomial time, so the final formula
has size $n^{O(\log k)}$, hence the $|\phi|^{O(1)}$ factor of the running time
contributes at most $n^{C\log k}$ for some constant $C$. Furthermore, according
to Lemma~\ref{lem:ksumtw} the CNF formula has primal treewidth at most
$\frac{k}{2}\log n + 4\delta k\log n +O(\log k)$.  Let $k_0$ be sufficiently
large that $(C+2)\log k_0 < \delta k_0$. Then, the total running time is at
most $2^{(1-\eps)(\frac{k}{2}\log n+4\delta k\log n+O(\log k))}n^{\delta k}$
which is at most $n^{\frac{k}{2}(1-\eps+10\delta)}$ and since $\delta$ is
sufficiently small we obtain the theorem.  \end{proof}

\section{Applications}\label{sec:applications}

As mentioned, the main application of the results of this paper is in the
fine-grained parameterized complexity of problems parameterized by treewidth.
In this area, numerous results are known proving that known dynamic programming
algorithms are optimal, under the \seth, or more recently the \ppseth.  In
particular, Lokshtanov, Marx, and Saurabh \cite{LokshtanovMS18} started a line
of work which succeeded in attacking problems for which an algorithm running in
$c^{\tw}n^{O(1)}$ was known, for $c$ a constant, and showing that an algorithm
with running time $(c-\eps)^{\tw}n^{O(1)}$ would falsify the \seth. The list of
problems for which such tight results are known is long and includes
\textsc{Independent Set}, \textsc{Dominating Set}, $k$-\textsc{Coloring},
\textsc{Max Cut} \cite{LokshtanovMS18}, \textsc{Steiner Tree}, \textsc{Feedback
Vertex Set} \cite{CyganNPPRW22}, \#-\textsc{Perfect Matching}
\cite{CurticapeanM16}, and numerous other problems
\cite{BorradaileL16,FockeMINSSW23,FockeMINSSW25,FockeMR22,KatsikarelisLP22,LampisV24,abs-2502-14161,MarxSS21}.
Analogous results are also known for other structural parameters, such as
clique-width and cut-width
\cite{BojikianCHK23,BojikianK24,KatsikarelisLP19,Lampis20}.

We will rely heavily on the work of Iwata and Yoshida who first put forward the
\twseth\ as a hypothesis (though not explicitly under this name)
\cite{IwataY15}. Their main result was the following:

\begin{theorem}\label{thm:iwata}

The following statements are equivalent:

\begin{enumerate}

\item (The \twseth\ is false): There exists $\eps>0$ and an algorithm solving
\textsc{SAT} in time $(2-\eps)^{\tw}|\phi|^{O(1)}$, where $\tw$ is the width of
a given tree decomposition of the primal graph of $\phi$.

\item The same as the previous statement, but for $3$-\textsc{SAT}.

\item The same as the previous statement, but for \textsc{Max-2-SAT}.

\item There exists an $\eps>0$ and an algorithm solving \textsc{Independent
Set} in time $(2-\eps)^{\tw}n^{O(1)}$.

\item There exists an $\eps>0$ and an algorithm solving \textsc{Independent
Set} in time $(2-\eps)^{\cw}n^{O(1)}$, where $\cw$ is the width of a given
clique-width expression.

\end{enumerate}

\end{theorem}

The main result of this paper is that the \ksum\ and \kxor\ Hypotheses imply
the \twseth. As a result we immediately obtain that the two hypotheses imply
that all the statements of Theorem~\ref{thm:iwata} are false. Going a bit
further, we prove the statements of Theorem~\ref{thm:applications} (restated
below).

\thmapps*

\begin{proof}

The first two statements follow immediately from Theorems~\ref{thm:kxor},
\ref{thm:ksum}, and the last two statements of Theorem~\ref{thm:iwata}.

The third statement follows from the same theorems and the third statement of
Theorem~\ref{thm:iwata} via the following simple reduction from
\textsc{Max-2-SAT} to \textsc{Max-Cut}.

\begin{claim} There is a treewidth-preserving reduction from \textsc{Max-2-SAT}
to \textsc{Max-Cut}. \end{claim}

\begin{claimproof} We describe a reduction that produces an edge-weighted
instance of \textsc{Max-Cut} where parallel edges are allowed. This can be
replaced by an unweighted instance by replacing edges of weight $w$ with $w$
parallel edges of weight $1$; and this can in turn be reduced to a simple graph
by replacing all edges with paths of length $3$ through new vertices and adding
$2m$ to the target cut value, where $m$ is the number of edges.  These
transformations do not affect the treewidth of the instance as long as the
initial graph has treewidth at least $2$ (which can be assumed without loss of
generality).

Suppose then that we are given a 2-CNF formula $\phi$ with $n$ variables and
$m$ clauses. We construct a graph that has one vertex for each variable, as
well as a special vertex $s_0$. For each clause $c=(\ell_1\lor \ell_2)$ we
construct two vertices $c_{\ell_1}$ and $c_{\ell_2}$ and make them adjacent to
each other and to $s_0$ via edges of weight $1$. Furthermore, if $\ell_1$ is a
positive appearance of variable $x$ we connect $c_{\ell_1}$ to the vertex
representing $x$ via a path of length $2$ through a new vertex and assign
weight $2m$ to its edges; while if $\ell_1$ is a negative appearance of $x$ we
connect $c_{\ell_1}$ to $x$ via an edge of weight $4m$ (and we do the same for
$\ell_2$). If the original instance had target value $t$ we set the target for
the new instance to $8m^2+2t$.

This completes the construction and we claim that the new graph has the same
treewidth as the primal graph of $\phi$ up to a small additive constant.
Indeed, we can take the tree decomposition of the primal graph and add $s_0$
everywhere. Then, for each clause $c=(\ell_1\lor \ell_2)$ involving variables
$x_1,x_2$ we observe that there must be a bag $B$ containing $x_1,x_2$; we make
a new bag which is adjacent to $B$ and contains
$x_1,x_2,s_0,c_{\ell_1},c_{\ell_2}$ and any other (degree $2$) vertices of the
gadget of clause $c$.  Proceeding in this way we obtain a valid tree
decomposition.

Now, if there is an assignment to $\phi$ satisfying $t$ clauses, we place in
the new graph all false variables and false literal vertices on the same side
as $s_0$ and all other variable and literal vertices on the other side. Observe
that this cut can be extended by appropriately assigning the remaining vertices
so that it cuts all edges of weight $2m$ and $4m$, which contributes $8m^2$ to
the cut.  Furthermore, if a clause is satisfied, at least one of the literal
vertices is on the other side of $s_0$, so the corresponding gadget contributes
$2$ more and we achieve the target cut size.

For the converse direction, observe that in a cut that achieves the target
value we must cut all edges of weight $2m$ or $4m$ because if one of them is
not cut the cut has weight at most $8m^2-2m+2m \le 8m^2$ which is below the
target. However, if all such edges are cut we can obtain an assignment to
$\phi$ by setting to $0$ all variables whose vertices are on the same side as
$s_0$ and this assignment is consistent with the assignment we obtain if we do
the same for literal vertices. Observe now that since for each clause we added
a triangle, each clause gadget contributes at most $2$ to the cut. Furthermore,
no clause gadget can contribute $1$ from its unweighted edges. Therefore, there
exist at least $t$ gadgets that contribute exactly $2$. But, this can happen
only if one of the literal variables is on a different side from $s_0$, that
is, if the corresponding clause is satisfied.  \end{claimproof}

Finally, to obtain the last statement on $k$-\textsc{Coloring} we make the
following claims, which reuse existing tools from the literature (in fact, we
simply verify that the details of previous reductions go through in our
setting):

\begin{claim} If the \twseth\ holds, then for all $\eps>0$ and $B\ge 3$,
2-\textsc{CSP} instances $\psi$ over alphabets of size $B$ cannot be solved in
time $(B-\eps)^{\tw}|\psi|^{O(1)}$. \end{claim}

\begin{claimproof} We recall that in \cite{Lampis25} the same statement was
shown for pathwidth rather than treewidth (Lemma 13 of \cite{Lampis25}). The
proof of Lemma 13 of \cite{Lampis25} in turn relies on Lemma 12 of
\cite{Lampis25}, which states that one can reduce the arity of a given
\textsc{CSP} instance from an arbitrary constant to $2$ without increasing the
pathwidth.

We observe that the proof of Lemma 12 easily goes through in our case, as it
relies on replacing a constraint of high arity with a cycle on new variables
and then attaching this cycle somewhere in the decomposition where all the
variables of the original constraint appear. This is still possible in a tree
decomposition (in fact, it is easier than in a path decomposition, as we can
attach a new branch to the corresponding bag).

The proof of Lemma 13 also goes through in our case. In particular, given a
decomposition of a CNF formula, we partition the boolean variables of each bag
in groups of $\delta$ variables so that each will be represented by $\gamma$
variables of the new \textsc{CSP} instance. We will have that $B^\gamma\sim
2^\delta$. The remaining construction is identical except we have consistency
constraints between any two neighboring bags of the decomposition, but it can
be shown that this gives the desired treewidth in the same way that these
constraints are shown to have the desired pathwidth in \cite{Lampis25}.
\end{claimproof}

\begin{claim} If there exist $\eps>0, k\ge 3$ such that $k$-\textsc{Coloring}
can be solved in time $(k-\eps)^{\tw}n^{O(1)}$, then there is an algorithm that
can solve 2-\textsc{CSP} instances $\psi$ over alphabets of size $k$ in time
$(k-\eps)^{\tw}|\psi|^{O(1)}$. \end{claim}

\begin{claimproof} This follows from the construction of Lemma 22 of
\cite{Lampis25}, where for each constraint of the 2-\textsc{CSP} we consider
each non-satisfying assignment and represent by a simple gadget which is
essentially a short path between the vertices representing the variables of the
constraint. This construction does not increase the pathwidth by more than an
additive constant, and it also clearly does not significantly increase the
treewidth either.  \end{claimproof}

We remark that, as in \cite{Lampis25}, it is not too hard to also obtain a
reduction in the converse direction (from $k$-\textsc{Coloring} parameterized
by treewidth to \twseth), but this is not necessary to establish
Theorem~\ref{thm:applications}.  \end{proof}



\newpage

\bibliography{ksum}

\end{document}